%% file: bicliquearxiv.tex
\documentclass{article}
\usepackage{a4wide}
\usepackage{amsmath}
\usepackage{amsthm}
\usepackage{amssymb}
\usepackage{tikz}
\usepackage[symbol]{footmisc}
\usepackage[ruled]{algorithm2e}

\newcommand{\FPT}{\mathbf{FPT}}

\newcommand{\Wone}{\mathbf{W[1]}}
\newcommand{\fpt}{\textup{fpt}} 
\newcommand{\biclique}{\ensuremath{\textsc{Biclique}}}
\newcommand{\CSP}{\ensuremath{\textsc{CSP}}}
\newcommand{\OCSP}{\ensuremath{\textsc{OCSP}}}
\newcommand{\CCSP}{\ensuremath{\textsc{CCSP}}}

\newcommand{\pclique}{\ensuremath{k\textsc{-Clique}}}

\newcommand{\kbiclique}{\textsc{$k$-Biclique}}
\newcommand{\kclique}{\textsc{$k$-Clique}}

\newcommand{\ETH}{\mathbf{ETH}}
\newcommand{\NP}{\mathbf{NP}}
\newcommand{\PCP}{\textup{PCP}}
\newcommand{\psubiso}{\textsc{$p$-Subgraph-Isomorphism$(\mathbf{C},-)$}}

\newlength{\probwidth}
\newcommand{\pagewidth}{10}

\newcommand{\npprob}[5][8]{
\begin{center}\normalfont\fbox{
\addtolength{\probwidth}{#1cm}\parbox{\probwidth}{\textsc{#2}\\\hspace*{1.5em}
\begin{tabular}[t]{rp{#1cm}}
\textit{Input:}&#3.\\\textit{Parameter:}&#4.\\\textit{Problem:}&#5
\end{tabular}}}
\end{center}}

\newtheorem{theo}{Theorem}[section]
\newtheorem{theorem}[theo]{Theorem}
\newtheorem{lemma}[theo]{Lemma}
\newtheorem{corollary}[theo]{Corollary}

\newtheorem{conjecture}[theo]{Conjecture}

\newtheorem{definition}[theo]{Definition}
\newtheorem{remark}[theo]{Remark}

\begin{document}

\title{The Parameterized Complexity of $k$-Biclique Problem}
\author{Bingkai Lin}

\maketitle

\begin{abstract}
Given a  graph $G$ and an integer $k$, the $\kbiclique$ problem asks whether $G$ contains a complete bipartite subgraph with $k$ vertices on its each side. 
Whether there is an $f(k)\cdot |G|^{O(1)}$-time algorithm  solving $\kbiclique$ for some computable function $f$ has been a longstanding open problem.

We show that $\kbiclique$ is $\Wone$-hard, which implies that such an $f(k)\cdot |G|^{O(1)}$-time algorithm does not exist under the hypothesis $\Wone\neq\FPT$ from  parameterized complexity theory.
To prove this result, we give a reduction which, for every $n$-vertex graph $G$ and  small integer $k$,  constructs a bipartite graph $H=(L\;\dot\cup\; R,E)$ in time polynomial in $n$ such that if $G$ contains a clique with $k$ vertices, then  there are $k(k-1)/2$ vertices in $L$ with $n^{\Theta(1/k)}$ common neighbors, otherwise any $k(k-1)/2$ vertices in $L$ have at most $(k+1)!$ common neighbors.
An additional feature of this reduction is that it creates a gap on the right side of the biclique. Such a gap might have further applications in proving hardness of approximation results.

Assuming a randomized version of Exponential Time Hypothesis, we establish an  $f(k)\cdot |G|^{o(\sqrt{k})}$-time lower bound for $\kbiclique$ for any computable function $f$. 
Combining our result with the work of~\cite{BulatovM11}, we obtain a dichotomy classification of the parameterized complexity of cardinality constraint satisfaction problems.
\end{abstract}

%
%

%


\section{Introduction}
Given an  $n$-vertex graph $G$ and an integer $k$,  the goal of $\kbiclique$  problem is to decide whether  $G$ contains a subgraph (not  necessarily induced) isomorphic to the balanced complete bipartite graph $K_{k,k}$. This problem is known to be $\NP$-hard \cite{johnson1987np}. Whether there exists an $f(k)\cdot n^{O(1)}$-time algorithm solving $\kbiclique$ for some computable function $f$ has received serious attention from the parameterized complexity community \cite{FlumGrohe06,Grohe07,BulatovM11}.
It is the first problem on the ``most infamous'' list (page 677) in a recent text book \cite{fundpara2013}:
\begin{quote}
``Almost everyone considers that this problem should obviously be $\Wone$-hard, and... it is rather an embarrassment to the field that the question remains open after all these years!''
\end{quote}
In this article, we confirm that $\kbiclique$ is $\Wone$-hard parameterized by $k$. Hence assuming $\Wone\neq\FPT$, a hypothesis from parameterized complexity theory  analogous to $\textup{NP}\neq \textup{P}$, it has no $f(k)\cdot n^{O(1)}$-time algorithms for any computable function $f$. As a byproduct, we also obtain a result on hardness of approximation in the parameterized setting.
We will explain this using the language of set intersection problem.
Given a collection $\mathcal{F}$ of subsets of $[n]:=\{1,2,\ldots,n\}$, the goal of \textsc{Maximum-$k$-Subset-Intersection} is to select $k$ distinct sets from $\mathcal F$ such that the size of their intersection  is as large as possible. 
 It is not difficult to see that $\kbiclique$ restricted to bipartite graphs can be interpreted as finding $k$ distinct vertices from one side of the bipartite graph
such that the intersection size of their neighbor sets is  at least $k$. 
Our result yields that one can construct a set family $\mathcal{F}$ and an integer $s$ for every graph $G$ and  integer $k$ in  time polynomial in $|G|$ such that $s=\binom{k}{2}$ and
\begin{description}
\item[(F1)] if $G$ contains a clique with $k$ vertices, then  there are  $s$ sets in $\mathcal F$ with intersection  size no less than $n^{\Theta(1/{k})}$,
\item[(F2)] if $G$ contains no clique with $k$ vertices, then any $s$ distinct sets from $\mathcal F$ have intersection size  at most $(k+1)!$.
\end{description}
We say an algorithm approximates \textsc{Maximum-$k$-Subset-Intersection} to a ratio $r\ge 1$ if it outputs $k$ sets from $\mathcal F$ whose intersection size is at least $1/r$ times the optimum one.
Combined with the $f(k)\cdot n^{o(k)}$-time lower bound for $\kclique$~\cite{chen2004linear} under the Exponential Time Hypothesis~\cite{impagliazzo1998problems}, our reduction implies that, assuming  $n$-variable \textup{SAT} has no $2^{o(n)}$-time algorithms,  there are no $f(k)\cdot n^{o(\sqrt{k})}$-time algorithms that can approximate the  \textsc{Maximum-$k$-Subset-Intersection} to ratio $n^{o(1/\sqrt{k})}$. On the other hand, assuming that \textup{SAT} cannot be solved by probabilistic algorithms in time $2^{n^\epsilon}$, the polynomial time inapproximability of \textsc{Maximum-$k$-Subset-Intersection} within ratio $n^{\epsilon'}$ ($\epsilon'$ depends on $\epsilon$) has been established in \cite{xavier2012note} basing on the inapproximability of \textsc{Maximum-Edge-Biclique} \cite{ambuhl2011inapproximability}. Note that our result refutes the existence of $f(k)\cdot n^{o(\sqrt{k})}$-time algorithms for any computable function $f$, while the work in~\cite{xavier2012note} only rules out the existence of polynomial-time algorithms.

The results in~\cite{ambuhl2011inapproximability} use the Quasi-random \PCP\ construction of ~\cite{khot2006ruling}. It is worth pointing out that the $\PCP$-machinery uses reduction which acts globally~\cite{arora1998approximability}. In contrast, our reduction performs local transformations, i.e., each bit of the output depends on at most constant bits of the input. The drawback of our inapproximability result is  that the gap we establish here is not so robust. In the case (F2), there may exist $s-1$ distinct sets in $\mathcal F$ with intersection size $n^{\Theta(1/{k})}$. In other words, we do not prove the hardness approximation of \textsc{Maximum-Balanced-Biclique} or \textsc{Maximum-Edge-Biclique}, whose inapproximability has been considered as major open problems in complexity theory (see~\cite{feige2002relations,feige2004hardness,khot2006ruling,ambuhl2011inapproximability}).

The main idea of our reduction is to exploit the gap between the sizes of the common neighbors of $k$-vertex sets and $(k+1)$-vertex sets in some Paley-type graphs defined in~\cite{BabaiGKRSW96}. Here we give a high level overview of the underlying idea  of our reduction using the language of  set intersection. First, suppose we can construct a set family $\mathcal{T}=\{S_1,S_2,\ldots,S_{n}\}$ of subsets of $[n]$ for some integers $k$, $n$ and $h>\ell$ (e.g. $h={n}^{1/k}$ and $\ell=(k+1)!$) such that:
\begin{description}
\item[(T1)] any $k+1$ distinct subsets in $\mathcal T$ have intersection size at most $\ell$;
\item[(T2')] any $k$ distinct subsets in $\mathcal T$ have intersection size at least $h$.\footnote{We will define later a property (T2) as a replacement for (T2').}
\end{description}
Then for every graph $G$ with $V(G)\subseteq [n]$, we  construct  our target set family  $\mathcal F$ by setting $\mathcal F:=\{S_{\{i, j\}} : \text{for all $\{i,j\}\in E(G)$}\}$, where $S_{\{i, j\}}:=S_i\cap S_j$. Let $s:=k(k-1)/2$. It is easy to check that if $G$ has a $k$-vertex clique, say $\{a_1,a_2,\ldots,a_k\}$ is a clique in $G$, then  (T2') implies that $|\bigcap_{i\in[k]}S_{a_i}|\ge h$. It follows that   $\{S_{\{a_i, a_j\}} : \text{for all $\{i,j\}\in\binom{[k]}{2}$}\}$ are $s$ distinct subsets in $\mathcal F$ with intersection size at least $h$. On the other hand, if $G$ contains no $k$-vertex clique, then any $s$ distinct sets in $\mathcal F$ must come from at least $k+1$ distinct sets in $\mathcal T$, by (T1) these sets have intersection size at most $\ell$.

To complete our reduction, it remains to construct the set family $\mathcal T$ efficiently for some appropriate parameters $k$, $n$, $h$ and $\ell$. However, at the moment of writing, I do not know how to do that even probabilistically. Therefore, we relax (T2') by  partitioning $[n]$ into $|V(G)|$ disjoint subsets $I_1,I_2,\ldots,I_{
|V(G)|}$ and replacing (T2') with (T2). 
\begin{description}
\item[(T2)] for any $k$ distinct vertices $a_1,a_2,\ldots,a_k\in  V(G)$, there exist $b_1\in I_{a_1},b_2\in I_{a_2},\ldots,b_k\in I_{a_k}$ such that $|\bigcap_{i\in[k]}S_{b_i}|\ge h$.
\end{description}
With a little more effort, we adapt our reduction to  set families satisfying (T1) and (T2). The  most technical part of this article is to construct  set families satisfying these two conditions for   $\ell=\Theta((k+1)!)<h={n}^{\Theta(1/k)}$. 
We also provide a probabilistic construction for  $\ell=\Theta(k^2)<h={n}^{\Theta(1/k)}$, which allows us to derive tighter lower bound for $\kbiclique$.

\subsection{Main Results}
In the theorems and corollaries of this section, $f$ can be any computable function.
\begin{theorem}[Main]\label{thm:gapreduction}
 For every  $n$-vertex graph $G$ and  positive integer $k$ 
there is a polynomial time algorithm which outputs a
 bipartite graph $H=(L\;\dot\cup\;R,E)$ and an integer $s=\Theta(k^2)$ such that:
\begin{description}
\item[(Completeness)] if  $G$ contains a clique with $k$ vertices, then there are $s$ vertices in $L$ with at least $n^{\Theta(\frac{1}{k})}$ common neighbors in $R$;
\item[(Soundness)] if  $G$ contains no clique with $k$ vertices, then any $s$ vertices in $L$ have at most $(k+6)!$ common neighbors in $R$.
\end{description}
\end{theorem}

From Theorem~\ref{thm:gapreduction}, we  obtain an inapproximation result for  \textsc{Maximum-$k$-Subset-Intersection} immediately.

\begin{corollary}
Assuming $\FPT\neq\Wone$, there is no $f(k)\cdot n^{O(1)}$-time algorithm approximating \textsc{Maximum-$k$-Subset-Intersection}  within  within $n^{o(\frac{1}{\sqrt{k}})}$-approximation ratio.
\end{corollary}

To see that Theorem \ref{thm:gapreduction} implies the $\Wone$-hardness of $\kbiclique$. Let $t=(k+6)!+1$. We add $(t-s)$ vertices to $H$ and make them adjacent to every vertex in $R$. It is easy to check that the resulting graph contains a $K_{t,t}$ if and only if the original graph $G$ contains a $K_k$.
\begin{corollary}
$\kbiclique$ is $\Wone$-hard.
\end{corollary}

\subsubsection{Hardness Results under $\ETH$}
More refined lower bounds can be obtained if we take a stronger assumption made by
Impagliazzo, Paturi and Zane \cite{impagliazzo1998problems,impagliazzo2001complexity}. 
\begin{conjecture}[Exponential Time Hypothesis (\textup{ETH})]
$3$-\textup{SAT} cannot be solved in time $2^{o(n)}$, where $n$ is the number of variables.
\end{conjecture}
\begin{theorem}[\cite{chen2004linear}]\label{thm:ETHkclique}
Assumming $\ETH$ there is no $f(k)\cdot n^{o(k)}$-time algorithm for $\kclique$.
\end{theorem}

From Theorem \ref{thm:ETHkclique} and Theorem \ref{thm:gapreduction}, we can deduce
\begin{corollary}
Assuming $\ETH$, there is no $f(k)\cdot n^{o(\sqrt{k})}$-time algorithm approximating \textsc{Maximum-$k$-Subset-Intersection}  within $n^{o(\frac{1}{\sqrt{k}})}$-approximation ratio.
\end{corollary}

An immediate open question is  whether there exist $f(k)\cdot n^{o(k)}$-time algorithms for $\kbiclique$. To rule out such algorithms, we need to find a \emph{linear} $\fpt$-reduction from $\kclique$ to $\kbiclique$, i.e., given $G$ and $k$, constructing a new graph $G'$ in $f(k)\cdot n^{O(1)}$ time such that $K_k\subseteq G$ if and only if $K_{k',k'}\subseteq G'$, where $k'=ck$ for some constant $c$. The existence of such a reduction would imply that $\kbiclique$ has no $f(k)\cdot n^{o(k)}$-time algorithm under  $\ETH$. However, since our reduction causes a quadratic blow-up of the size of solution, $k'=\binom{k}{2}$ is the best we can achieve so far. We note that by Theorem \ref{thm:gapreduction}, we can get $k'=\Omega(k!)$. Nevertheless, using the probabilistic method, we have:

\begin{theorem}\label{thm:probbicw1}
For every $n$-vertex graph $G$ and  positive integers $k$, $\ell$ and $h$ with $n\ge \max\{4(k+1)^2,20\}$, $\ell=2k^2+4k-1<h\le n^{\frac{1}{4(k+1)}}$, one can construct a random graph $H=(L\;\dot\cup \;R,E)$ in time polynomial in $n$  such that, with  probability at least $\frac{9}{10}$, 
\begin{description}
\item[(Completeness)]  if $G$ contains a $k$-clique, then there exists $\binom{k}{2}$ vertices in $L$ having $h$ common neighbors,
\item[(Soundness)] if $G$ contains no $k$-clique, then every $\binom{k}{2}$-vertex in $L$ has at most $\ell$ common neighbors.
\end{description}
\end{theorem}

Consider a randomized version of $\ETH$ which states that there are no randomized algorithms with two-sided error such that
for every input instance of $3$-\textup{SAT} decide if it is satisfiable or not  correctly with  probability larger than $1/2$ in $2^{o(n)}$ time. For more detail we refer to  \cite{ParameterizedAlgorithm}. With the randomized $\ETH$ Theorem~\ref{thm:probbicw1} yields a better lower bound for $\kbiclique$:
\begin{corollary} Under the randomized $\ETH$,
there is no $f(k)\cdot n^{o(\sqrt{k})}$-time algorithm to  decide whether a given graph contains a subgraph isomorphic to $K_{k,k}$.
\end{corollary}

\subsection{Related Topics}

\paragraph{Cardinality $\CSP$} $\kbiclique$  can be formulated under the framework of constraint satisfaction problem. Fix a domain  $D$, an instance of the constraint satisfaction problem ($\CSP$) is a pair $I=(V,C)$, where $V$ is a set of variables and $C$ is a set of constraints.  Each constraint of $C$ can be written as $\langle \textbf{v},R\rangle$, where $R$ is an $r$-ary relation on $D$ for some positive integer $r$ and $\textbf{v}=(v_1,v_2,\ldots, v_r)$ is an $r$-tuple of variables. An assignment $\tau: V\to D$ satisfies a constraint $\langle \textbf{v},R\rangle$ if and only if $(\tau(v_1),\tau(v_2),\ldots,\tau(v_r))\in R$. The goal is to find an assignment $\tau: V\to D$ satisfying all the constraints in $C$. In many applications of $\CSP$, we usually fix a set $\Gamma$ of relations, and denote by $\CSP(\Gamma)$ the $\CSP$ problem in which all  the relations in the constraints are from $\Gamma$.

In \cite{BulatovM11}, Andrei A. Bulatov and D{\'a}niel Marx introduced two parameterized versions of $\CSP$. More specifically, they assume that the domain contains a special value $0$, which is ``free'', and other non-zero values, which are ``expensive''. The goal is to find an assignment with bounded number of variables assigning expensive values. One way to achieve this goal is to take the number of nonzero values used in an assignment as parameter, which leads to the definition of the $\CSP$ with size constraints ($\OCSP$); another more refined way is to prescribe how many variables have to be assigned each particular nonzero value, this leads to the definition of $\CSP$ with cardinality constraints. They provide a complete characterization of the fixed-parameter  tractable cases of $\OCSP(\Gamma)$ and show that all the remaining problems are $\Wone$-hard.

For $\CSP$ with cardinality constraints, the situation was more complicated. A simple observation shows that  $\kbiclique$  can be expressed as a $\CCSP$ instance. Without loss of generality, consider $\kbiclique$ on bipartite graphs. Let $D:=\{0,1,2\}$. For any bipartite graph $G=(A\;\dot\cup\;B,E)$, we construct a $\CCSP$ instance with $V=A\;\dot\cup\;B$ and 
\[
C=\{\langle (v_1,v_2),R\rangle : \text{for all $v_1\in A$ and $v_2\in B$ with  $v_1v_2\in E$ and $R=\{(0,0),(1,0),(0,2)\}$}\},
\]
then we ask for an assignment $\tau: V\to D$ with $k$ variables assigning $1$ and $k$ variables assigning $2$. It is easy to check that for any bipartite graph $G$,  the corresponding $\CCSP$ instance has such an assignment if and only if the \emph{bipartite complement}\footnote{The bipartite complement of $G=(A\;\dot\cup\; B,E)$ is defined by $\bar{G}:=(A\;\dot\cup\; B,\bar{E})$, where $\bar{E}:=\{\{v,u\} : v\in A,u\in B,\{v,u\}\notin E\}$.} $\bar{G}$ of $G$ contains a $K_{k,k}$. Therefore, without settling  the parameterized complexity of $\kbiclique$, they can only show that $\CCSP(\Gamma)$ is  fixed-parameter tractable, $\biclique$-hard or $\Wone$-hard. Combining our result and Theorem 1.2 in  \cite{BulatovM11}, we finally obtain a dichotomy theorem for the parameterized complexity of $\CCSP(\Gamma)$:
\begin{theorem}
For every finite $\Gamma$ closed under substitution of constants, $\CCSP(\Gamma)$ is either $\FPT$ or $\Wone$-hard.
\end{theorem}

\paragraph{Subgraph Isomorphism}
$\kbiclique$ is a special case of \textsc{Subgraph-Isomorphism}, in which we are given two graphs $G$ and $H$ and the goal is to decide if $G$ is a subgraph of $H$. Using the color-coding technique in \cite{AlonYZ95},  $\textsc{Subgraph-Isomorphism}$ can be solved in $2^{O(|G|)}\cdot|H|^{O(tw(G))}$ time, where $tw(G)$ denotes the tree-width of $G$. (For the definition of tree-width, see~\cite{robertson1986graph,kloks1994treewidth,bodlaender1994tourist}.) For any class $\mathbf{C}$ of graphs, define the following parameterized problem.
\npprob[7]{\psubiso}{A graph $G\in \mathbf{C}$ and an arbitrary graph $H$}{$|G|$}{Decide whether $G$ is a subgraph of $H$.}
It follows that if $\mathbf{C}$ is a class of graphs with tree-width bounded by some constant, then \psubiso\ is fixed-parameter tractable, and this is believed to be optimal.
In \cite{Grohe07}, Grohe conjectured  that \psubiso\ is $\Wone$-hard if and only if $\mathbf{C}$ has unbounded tree-width. Under the assumption of $\FPT\neq \Wone$, this would imply that there is no $f(k)\cdot |H|^{O(1)}$-time algorithm to decide whether $H$ contains a subgraph isomorphic to $K_{k,k}$, because the class of balanced complete bipartite graphs $\{K_{k,k} : k\in\mathbb{N}\}$ has unbounded tree-width. In other words, we can not prove Grohe's conjecture without answering the parameterized complexity of $\kbiclique$. Although $\kbiclique$ is believed to be $\Wone$-hard, despite many attempts  \cite{marx2007can,binkele2010exact,couturier2012bicolored,AtminasLR12,gaspers2012independent,kutzkov2012exact}, no $\fpt$-reduction from $\kclique$ to $\kbiclique$ has previously been found. So proving the $\Wone$-hardness of $\kbiclique$ is one step towards a dichotomy classification of \psubiso.

\bigskip
\noindent \textbf{Organization of the Paper.} In Section~\ref{sec:preliminary} we set up notations and terminologies.
The main idea of the reduction is presented in Section~\ref{sec:reduction} after introducing the threshold property. To complete the reduction, we  provide two efficient constructions of  bipartite graphs satisfying the threshold property. The  explicit construction is given in Section~\ref{sec:paley-type}, while  a probabilistic construction  can be found in Section~\ref{sec:probaconst}. Finally, we conclude and raise some open questions in Section~\ref{sec:conclusion}.

\section{Preliminaries}\label{sec:preliminary}
We use $\mathbb N$, $\mathbb N^+$, $\mathbb{R}$ and  $\mathbb C$ to denote the sets of nonnegative integers, positive integers, real numbers and complex numbers respectively.
For any number $n\in\mathbb{N}^+$, let $[n] := \{1,2, \ldots, n \}$. For $d,n\in\mathbb{N}^+$, we write $d \mid n$ if and only if $d$ divides $n$. For any real numbers $a,b$, we use the notation $a\pm b$ to denote the numbers between $a-b$ and $a+b$. 
 For every set $S$ we use $|S|$ to denote its size. Moreover, for any $t\in\mathbb{N}^+$, we let
$\binom{S}{t}$ be the set of all $t$-element subsets of $S$.  We use $A\;\dot\cup\; B$ to denote the  union of two disjoint sets $A$ and $B$.

\subsection{Parameterized Complexity}
We denote the alphabet $\{0,1\}$ by $\Sigma$ and identify problems $Q$
with subsets  of $\Sigma^*$.
A \emph{parameterized problem} is a pair $(Q,\kappa)$ consisting of
a classical problem $Q \subseteq \Sigma^*$ and a polynomial time
computable {\em parameterization} $\kappa : \Sigma^* \to \mathbb N$.
For example, the parameterized clique problem is defined in the form:
\npprob[\pagewidth]{\pclique}{A  graph $G$ and a positive integer $k$}{$k$}{Does $G$ contains a subgraph isomorphic to $K_{k}$?}

An algorithm $\mathbb A$ is an \emph{$\fpt$-algorithm with respect
to a parameterization $\kappa$ } if for every $x\in \Sigma^*$ the
running time of $\mathbb A$ on $x$ is bounded by $f(\kappa(x))\cdot
|x|^{O(1)}$ for a computable function $f:\mathbb N\to \mathbb N$.
A parameterized problem is fixed-parameter tractable ($\FPT$ for short) if it has an $\fpt$-algorithm.

Let $(Q,\kappa)$ and $(Q',\kappa')$ be two parameterized problems.
An \emph{\fpt-reduction} from $(Q,\kappa)$ to $(Q',\kappa')$ is a
mapping $R:\Sigma^*\to \Sigma^*$ such that:
\begin{description}

\item For every $x\in \Sigma^*$ we have $x\in Q$ if and only if
    $R(x)\in Q'$;

\item $R$ is computable by an \fpt-algorithm with respect to $k$;

\item There is a computable function $g:\mathbb N\to \mathbb N$
    such that $\kappa'(R(x)) \le g(\kappa(x))$ for all $x\in
    \Sigma^*$.

\end{description}
If there is an \fpt-reduction from $(Q,\kappa)$ to $(Q',\kappa')$ and  $(Q',\kappa')$ is $\FPT$,  then $(Q,\kappa)$ is also $\FPT$. On the other hand, 
if there is an \fpt-reduction from $\pclique$ to $(Q',\kappa')$,  then we say $(Q',\kappa')$ is  $\Wone$-hard  (for the definition of $\Wone$-hardness, see \cite{dowfel,FlumGrohe06}). A standard assumption from parameterized complexity theory is $\FPT\neq\Wone$, which is equivalent to stating that $\pclique$ has no $f(k)\cdot |G|^{O(1)}$-time algorithm for any computable function $f$.

\subsection{Graphs}
Every graph $G= (V,E)$ is  determined by a nonempty vertex set $V$ and an edge set $E\subseteq \binom{V}{2}$. We also use $V(G)$ and $E(G)$ to denote the vertex set and edge set of $G$. Sometimes, an edge $\{u,v\}$  is written as $uv$. Let $|G|:=|V(G)|+|E(G)|$ be the size of  $G$.  Every nonempty subset $S\subseteq V(G)$ induces a subgraph $G[S]$ with the vertex set $S$ and the edge set $E(G[S]):= \binom{S}{2}\cap E(G)$. 
A graph $G$ is a subgraph of another graph $H$, denoted by $G\subseteq H$, if there exists an injection $\phi : V(G)\to V(H)$ such that for all $u,v\in V(G)$, $\{u,v\}\in E(G)$ implies that $\{\phi(u),\phi(v)\}\in E(H)$. We call $G$  a \emph{clique}  if for every distinct $u,v \in V(G)$ we have $\{u,v\}\in E(G)$. 
A $k$-clique is a clique with exactly $k$ vertices, which sometimes we also call a $K_k$ to simplify presentation. 
We call $G$ \emph{bipartite} if $V(G)$ admits a partition into two classes such that every edge has its ends in different classes.  A \emph{complete bipartite graph} or \emph{biclique} is a bipartite graph such that every two vertices from different partition classes are adjacent. We use  \emph{$K_{s,t}$} to denote the complete bipartite graph with $s$ vertices on one side and $t$ vertices on the other side. In the bipartite graph $G=(A\;\dot\cup \;B,E)$, for $V\subseteq A$, let $\Gamma^G(V):=\{u\in B : \text{for all $v\in V, \{v,u\}\in E$}\}$. Similarly, for any $t$-tuple $\textbf{v}=(v_1,v_2,\ldots,v_t)$ of $A$, $\Gamma^G(\textbf{v}):=\{u\in B : \text{for all $i\in[t], \{v_i,u\}\in E$}\}$. If it is clear from the context, we omit the superscript $G$ in the above notations and write $\Gamma(V)$ and $\Gamma(\textbf{v})$ instead.
\subsection{Probability}
In this paper we consider probability spaces $(\Omega,\Pr)$, where $\Omega$ is a finite set and $\Pr$ is a nonnegative function from $\Omega$ to $[0,1]$ such that $\sum_{\omega\in\Omega}\Pr(\omega)=1$. 

A \emph{random variable} over $(\Omega,\Pr)$ is a function $X : \Omega\to\mathbb{R}$.
The \emph{expectation} and \emph{variance}
of a random variable are defined as follows.
\begin{description}

\item[Expectation] $\mathrm{E}[X]:=\sum_{\omega\in\Omega}X(\omega)\cdot\Pr(\omega)$.
\item[Variance] $\mathrm{Var}[X]:=\mathrm{E}[X^2]-\mathrm{E}[X]^2$.

\end{description}
An \emph{event} $E$ can be treated as a subset of $\Omega$. The probability of $E$ is defined as $\Pr(E):=\sum_{\omega\in E}\Pr(\omega)$.
It is routine to define an event from a random variable. For example, given a random variable $X$, the event ``$X>0$'' can be regarded as a set $E:=\{\omega\in\Omega : X(\omega)>0\}$. Hence $\Pr(X>0)=\sum_{\omega\in\Omega, X(\omega)>0}\Pr(\omega)$.

To give upper bounds for probabilities in the forms $\Pr(X>\alpha)$ and $\Pr(X<\alpha)$ for random variable $X$ and $\alpha\in\mathbb{R}$, we need the following tools.
\begin{theorem}[Markov's Inequality]
Let $X\ge 0$ be a random variable and $\alpha>0$, then 
\begin{equation*}
\Pr(X\ge\alpha)\le \frac{\mathrm{E}[X]}{\alpha}.
\end{equation*}

\end{theorem}

\begin{theorem}[Chebyshev's Inequality]
For any real $\lambda>0$,
\begin{equation*}
\Pr(|X-\mathrm{E}[X]|\ge\lambda)\le \frac{\mathrm{Var}[X]}{\lambda^2}.
\end{equation*}
\end{theorem}

\subsection{Finite Fields}
We collect some useful facts on finite fields here. For more detail, the reader is referred to~\cite{lidl1997finite}. For every prime number $p$ the set $\{0,1,2,\ldots,p-1\}$ with addition and multiplication modulo $p$  is a  finite field, which we denote by $\mathbb{F}_p$.
A polynomial over a field $\mathbb{F}$ is irreducible  if it can not be expressed as the product of two non-constant polynomials over $\mathbb{F}$  with lower degree.
For every positive integer $t$, there always exists 
 an irreducible polynomial  $f$  with degree $t$ over the field $\mathbb{F}_p$. 
The  finite field $\mathbb{F}_{q}$ with $q=p^t$ elements can be represented as the set of polynomials over $\mathbb{F}_p$ with degree at most $t-1$. The addition and multiplication over $\mathbb{F}_{q}$ is performed modulo $f$. 
Let $\mathbb{F}_q^\times:=\mathbb{F}_q\setminus\{0\}$ be  the multiplicative group of $\mathbb{F}_q$. There always exists a generator $g\in \mathbb{F}_q^\times$ such that  $\mathbb{F}_q^\times=\{g^i : i\in[q-1]\}$.
The algebraic closure of any finite field $\mathbb{F}_{p^t}$ ($t\ge 1$) is the union $\bigcup_{i\ge 1}\mathbb{F}_{p^i}$. 
For $a\in\mathbb{N}$ and $e\in\mathbb{F}_{p^t}$, we denote by $a\cdot e$ the sum of $a$ copies of $e$. 
\begin{lemma}\label{lem:orderp}
For  $e\in\mathbb{F}_{p^t}$ and $n$ divisible by $p$,
$n\cdot e=0$.
\end{lemma}
Suppose $f(x)=a_0+a_1x+\cdots+a_nx^n$ is a polynomial over a field $\mathbb{F}$. The derivative $f'$ of $f$ is defined by $f'(x)=a_1+2\cdot a_2x+\cdots+n\cdot a_nx^{n-1}$.
\begin{theorem}\label{thm:multiplicity}
The element $b\in\mathbb{F}$ is a multiple root of $f$ if and only if it is a root of both $f$ and $f'$.
\end{theorem}

We use the following bounds in~\cite{shparlinski2013finite} for computations on finite fields. Let $L(n)=\log n\log\log n$.
\begin{theorem}\label{thm:A}
Let $\mathbb{F}$ be a field, $f$ and $g$ be two polynomials over $\mathbb{F}$ of degree at most $n$. Then $fg$ and the remainder of $f$ divided by $g$ can be computed using $O(nL(n))$ arithmetic operations in $\mathbb{F}$.
\end{theorem}
\begin{theorem}\label{thm:B}
For every prime $p$, addition in $\mathbb{F}_p$ can be performed using $O(\log p)$ bit operations; multiplication can be performed using $O(\log p L(\log p))$ bit operations and division can be performed using $O(L(p)L(\log p))$ bit operations.
\end{theorem}

To construct $\mathbb{F}_{p^t}$ we need to find  an  irreducible polynomial $f$ over $\mathbb{F}_p$ with degree $t$. In \cite{adleman1986finding}, the authors give a deterministic algorithm that outputs an irreducible polynomial with degree $t$ in  $(t\cdot\log p)^{O(1)}$ steps, assuming extended Riemann hypothesis. To remove the need for extended Riemann hypothesis in this algorithm is still an open question. However, in this paper we are allowed to use reductions running in time $q^{O(1)}$. A brute-force search algorithm will do the job. 
We enumerate every polynomial with degree  $t$ and test if it is divisible by any other non-constant polynomial with lower degree.
By Theorem~\ref{thm:A} and Theorem~\ref{thm:B}, the running time of this brute-force search algorithm can be upper bounded by $p^{O(t)}$.



\section{Reduction}\label{sec:reduction}
We start with a definition of graph property that is key to our reduction. We then present the reduction assuming the existence of graphs satisfying this property. We prove our main result via a deterministic construction of such graphs, and the $\ETH$-hardness result via a randomized construction.

\begin{definition}[$(n,k,\ell,h)$-threshold property]\label{def:thresholdproperty}
For $n,k,h,\ell\in\mathbb{N}^+$ with $h>\ell$, a bipartite graph $T=(A\;\dot\cup\; B,E)$ with a partition $A=V_1\;\dot\cup \;V_2\;\dot\cup\cdots\dot\cup\; V_n$  satisfies the $(n,k,\ell,h)$-threshold property if the following two conditions hold.
\begin{description}
\item[(T1)] Any $k+1$ distinct vertices in $A$ have at most $\ell$ common neighbors in $B$, i.e.
\[\forall V\in\binom{A}{k+1},|\Gamma(V)|\le \ell.\]
\item[(T2)] For all $k$ distinct indices $\{i_1,i_2,\ldots,i_k\}\in\binom{[n]}{k}$, there exist $v_{i_1}\in V_{i_1},v_{i_2}\in V_{i_2},\ldots,v_{i_k}\in V_{i_k}$ such that $v_{i_1},v_{i_2},\ldots,v_{i_k}$ have at least $h$ common neighbors in $B$, i.e.
\[\forall \{i_1,i_2,\ldots,i_k\}\in\binom{[n]}{k},\exists \textbf{v}\in V_{i_1}\times V_{i_2}\times \cdots\times V_{i_k},|\Gamma(\textbf{v})|\ge h.\]
\end{description}
\end{definition}

\begin{remark}
The definition of threshold property  was inspired by the work from a remarkable paper \cite{BabaiGKRSW96}, in which the authors gave an explicit construction  of a class of graphs satisfying (T1) and
\begin{description}
\item[($\star$)] At least a $\frac{1}{2\ell-1}$ fraction of the sets $V\in\binom{A}{k}$ have $h$ common neighbors,
\end{description}
for $\ell=(k+1)!$ and $h=n^{\Theta(1/k)}$. (see Theorem 3.6 and Lemma 3.7 of \cite{BabaiGKRSW96}) We replace the property ($\star$) by (T2) because in the reduction we need to ensure  that for every  $k$-vertex set of $G$, there exists a corresponding $k$-vertex set with $h$ common neighbors in $T$.
\end{remark}

\begin{lemma}[reduction]\label{lem:thresholdimplyhard}
Suppose $k,n\in\mathbb{N}^+$. Let $s:=\binom{k}{2}$. For every $n$-vertex simple graph $G$ and a bipartite graph $T$ satisfying the $(n,k,\ell,h)$-threshold property, one can construct a new graph $H=(L\;\dot\cup \;R,E)$ in $O(|V(T)|^2\cdot |G|)$ time, such that:
\begin{description}
\item[(H1)] if $G$ contains a clique with $k$ vertices, then there exists $V\in\binom{L}{s}$,  $|\Gamma^H(V)|\ge h$;
\item[(H2)] if $G$ contains no clique with $k$ vertices, then for all $V\in\binom{L}{s}$,  $|\Gamma^H(V)|\le \ell$.
\end{description}
\end{lemma}
\begin{figure}
\input{reductionfig}\caption{An example of the reduction when $k=3$ and $s=\binom{k}{2}=3$.  
In the (yes) case, $G$ contains a $k$-clique with the vertex set $\{a,b,c\}$. By (T2), there exist $u_a\in V_a$, $u_b\in V_b$ and $u_c\in V_c$ such that $\{u_a,u_b,u_c\}$ has at least $h$ common neighbors in $T$. Thus the $s$-verex set $\{u_au_b,u_au_c,u_bu_c\}$ has at least $h$ common neighbors in $H$; In the (no) case, $G$ does not contain any $k$-clique. Therefore, every subgraph of $G$ with  $s$ edges must have at least $k+1$ vertices. As a consequence,  every $s$-vertex  set in $L$ is constructed from  $(k+1)$ vertices in $A$ and thus has at most $l$ common neighbors by (T1). 
}\label{Fig:reduction}
\end{figure}
\begin{proof}
Suppose $V(G)=[n]$,  $T=(A\;\dot\cup\;B,E(T))$ and $A=V_1\;\dot\cup \;V_2\;\dot\cup\cdots\dot\cup \;V_n$. Our goal is to construct a bipartite graph $H=(L\;\dot\cup \;R,E)$ satisfying (H1) and (H2). We associate with each $V_i$ a vertex $i\in V(G)$. Let $\iota:A\to V(G)$ be the function that for each $i\in[n]$ and $u\in V_i$, $\iota(u)=i$. Then we construct the bipartite graph  $H=(L\;\dot\cup \;R,E)$ as follows.
\begin{description}
\item $L:=\{\{u_1,u_2\} : u_1,u_2\in A,\{\iota(u_1),\iota(u_2)\}\in E(G)\}$.
\item $R:=B$.
\item $E:=\{\{e,v\} : e=\{u_1,u_2\}\in L,v\in R,\{u_1,v\}\in E(T),\{u_2,v\}\in E(T)\}$.
\end{description}

Obviously, $H$ can be constructed in time $|V(T)|^2\cdot |G|$.
We only need to show that $H$ satisfies (H1) and (H2):
\begin{itemize}
\item[-] If  $G$ contains a $k$-vertex clique, we can assume that $\{a_1,a_2,\ldots,a_k\}\subseteq V(G)$ induces a $K_k$ in $G$. By (T2), for all $i\in [k]$ there exists $u_{a_i}\in V_{a_i}$ such that $\{u_{a_1},u_{a_2},\ldots,u_{a_k}\}$ has at least $h$ common neighbors in $B$. Let $X:=\{u_{a_1},u_{a_2},\ldots,u_{a_k}\}$ and $Y:=\Gamma^T(X)$. By definition, we have $|X|=k$ and $|Y|\ge h$. Let $E_X:=\binom{X}{2}$. Since $\{\iota(u_{a_i}),\iota(u_{a_j})\}=\{{a_i},{a_j}\}\in E(G)$ for all distinct $i,j\in [k]$, we have $E_X\subseteq L$. Hence for all $e\in E_X\subseteq L$ and $v\in Y\subseteq R$, $\{e,v\}\in E$. So $E_X\;\dot\cup\; Y$ induces a complete bipartite subgraph in $H$. It follows that $H$ satisfies (H1) because $|E_X|=\binom{|X|}{2}=\binom{k}{2}=s$ and $|Y|\ge h$.
\item[-] Suppose there exists $V\in\binom{L}{s}$ with $|\Gamma^H(V)|\ge \ell+1$, we must show that $G$ contains a clique with $k$ vertices. 
Let  $X:=\{u\in A : \text{there exists  $e\in V$ such that $u\in e$}\}$ and $Y:=\Gamma^H(V)$. By the definition of the edge set $E$,   we must have  $Y\subseteq \Gamma^T(X)$. Since $|Y|\ge \ell+1$ and for all $X'\in \binom{A}{k+1}$  $|\Gamma^T(X')|\le \ell$, we deduce that  $|X|\le k$; on the other hand, it is not hard to see that $V\subseteq\binom{X}{2}$, hence $|V|=\binom{k}{2}$ implies that $|X|>k-1$. Thus $|X|=k$ and for any distinct $u_1,u_2\in X$, $\{u_1,u_2\}\in V\subseteq L$. According to the definition of $L$,  for all  $u_1,u_2\in X$, $\{u_1,u_2\}\in  L$ implies that $\{\iota(u_1),\iota(u_2)\}\in E(G)$. Since $G$ is simple, $\{\iota(u) : u\in X\}$ must induce a $K_k$ in $G$.
\end{itemize}
\end{proof}

By Lemma~\ref{lem:thresholdimplyhard}, to prove Theorem~\ref{thm:gapreduction}, we only need to construct graphs satisfying the threshold property efficiently. Our main technical lemma is:

\begin{lemma}\label{lem:conthres}
For $k,n\in\mathbb{N}^+$ with $k=6\ell-1$ for some $\ell\in\mathbb{N}^+$ and $\lceil(n+1)^{\frac{6}{k+1}}\rceil>(k+1)!$, a bipartite graph $T=(A\;\dot\cup\; B,E)$ with $A=V_1\;\dot\cup \;V_2\; \dot\cup\cdots\dot\cup\; V_n$ satisfying the $(n,k,(k+1)!,\lceil(n+1)^{\frac{6}{k+1}}\rceil)$-threshold property and $|V(T)|=O(n^6)$ can be constructed  in $n^{O(1)}$ time. 
\end{lemma}
We postpone the proof of Lemma~\ref{lem:conthres} till Section~\ref{sec:paley-type}. Now we are ready to prove   Theorem~\ref{thm:gapreduction}.

\begin{theorem}[Theorem~\ref{thm:gapreduction} reformulated]
 For every  $n$-vertex graph $G$ and  positive integer $k$ with $n^{\frac{6}{k+6}}>(k+6)!$ 
there is a polynomial time algorithm constructing a
 bipartite graph $H=(L\;\dot\cup\;R,E)$ such that:
\begin{description}
\item[(Soundness)] if  $G$ contains a clique with $k$ vertices, then there are $s$ vertices in $L$ with at least $n^{\frac{6}{k+6}}$ common neighbors in $R$;
\item[(Completeness)] if  $G$ contains no clique with $k$ vertices, then any $s$ vertices in $L$ have at most $(k+6)!$ common neighbors in $R$,
\end{description}
where  $s=\binom{k'}{2}$ and $k'$ is the minimum integer such that  $k'+1$ is divisible by $6$ and $k'\ge k$.
\end{theorem}
\begin{proof}
We add a new clique with $k'-k$  vertices into $G$ and connect them with every vertex in $G$. It is easy to see that the new graph contains a $k'$-clique if and only if $G$ contains a $k$-clique. Since $ n^{\frac{6}{k+6}}>(k+6)!$ and $k+6\ge k'+1$, we have $\lceil n^{\frac{6}{k'+1}}\rceil\ge n^{\frac{6}{k+6}}>(k+6)!\ge (k'+1)!$. Apply Lemma~\ref{lem:conthres} on $n$ and $k'$, we obtain a graph with the $(n,k',(k'+1)!,\lceil(n+1)^{\frac{6}{k'+1}}\rceil)$-threshold property. The result then follows from Lemma~\ref{lem:thresholdimplyhard}.
\end{proof}

Using probabilistic method, we show:
\begin{lemma}\label{lem:conprob}
For  $k,\ell,h,n\in\mathbb{N}$ with $n\ge \max\{4(k+1)^2,20\}$, $\ell=2k^2+4k-1$ and $\ell<h\le n^{\frac{1}{4(k+1)}}$, one can construct in time polynomial in $n$ a bipartite random graph satisfying the $(n,k,\ell,h)$-threshold property with probability at least $\frac{9}{10}$.
\end{lemma}
Theorem~\ref{thm:probbicw1} then follows from Lemma~\ref{lem:conprob} and Lemma~\ref{lem:thresholdimplyhard}.

\section{Explicit Construction}\label{sec:paley-type}
In this section we give a deterministic construction of graphs satisfying the conditions (T1) and (T2) in Definition~\ref{def:thresholdproperty}. To that end, we need the definition of Paley-type bipartite graphs from~\cite{BabaiGKRSW96}. For certain parameter values, it is already known that such graphs satisfying (T1). As for (T2), we first give a partition of the vertices on the left side of the Paley-type bipartite graphs. Then we estimate the size of intersection  of any subset from this partition and the set of common neighbors  of any $k$-vertex set from the right side using Lemma~\ref{lem:manysolution}. In Lemma~\ref{lem:partitionofneighbor}, we then use a special case of Lemma~\ref{lem:manysolution} (when $k=1$) to prove  Paley-type  bipartite graphs with proper parameters  also satisfy (T2). Finally, we set up the parameters and prove Lemma~\ref{lem:conthres}.

\begin{definition}[Paley-type bipartite graph]
For any prime power $q$ and integer $d\mid q-1$, the Paley-type bipartite graph $P(q,d):=(A\;\dot\cup\; B,E)$ is defined as follows.
\begin{description}
\item[Vertices] $A=B=\mathbb{F}_q^\times$.
\item[Edges] For all $x\in A$ and $y\in B$, $xy\in E\iff (x+y)^{\frac{q-1}{d}}=1$.
\end{description}
\end{definition}

The Paley-type graphs have many nice properties, the following one is proved in \cite{kollar1996norm,BabaiGKRSW96}.
\begin{theorem}[Theorem 5.1 in \cite{BabaiGKRSW96}]\label{thm:bicliquefree}
For every prime power $q$ and integer $t>1$,
the graph $P(q^t,q-1)$ contains no subgraph isomorphic to $K_{t,t!+1}$.
\end{theorem}
Therefore, for any prime $p$ the graph $P(p^t,p-1)$
 satisfies (T1) for $k\gets t-1$ and $\ell\gets t!$, our next step is to show that it also satisfies (T2) for  appropriate choices of parameter $h$ and partition of the vertex set $A$. We need the following lemma.

\begin{lemma}[Intersection]\label{lem:manysolution}
For any $d,k,r,s\in\mathbb{N}^+$ and  prime power $q$ with $q-1=rs$, $d\mid (q-1)$ and $\sqrt{q}\ge\frac{sk}{d}+1$. Let $a_1,a_2,\ldots,a_k$ be $k$ distinct elements in $\mathbb{F}^{\times}_{q}$, $g$ be the generator of $\mathbb{F}^\times_{q}$. For each $j\in [s]$, let $V_j:=\{g^{j+s},g^{j+2s},\ldots,g^{j+sr}\}$. Then for any $j\in [s]$, the number of solutions $x\in V_j$ to the system of equations $(a_i+x)^{\frac{q-1}{d}}=1 (\forall i\in[k])$ is in
$\frac{q}{sd^k}\pm k\sqrt{q}$.
\end{lemma}

Lemma~\ref{lem:manysolution} generalizes Lemma 3.8 in \cite{BabaiGKRSW96} by restricting the solutions to any subset $V_j (j\in[s])$. If we set $s:=1$, then we obtain Lemma 3.8 in \cite{BabaiGKRSW96}.
The intuition behind Lemma~\ref{lem:manysolution} is that the solutions of $(a_i+x)^{\frac{q-1}{d}}=1$ distribute ``randomly'':  the equation $(a_i+x)^{\frac{q-1}{d}}=1$ has $\frac{q-1}{d}$ solutions, we may say that a random generated element $x\in \mathbb{F}^\times_{q}$ satisfies this equation with probability $\frac{1}{d}$, hence $x$ satisfies the system of equations $(a_i+x)^{\frac{q-1}{d}}=1 (\forall i\in[k])$ with probability $\frac{1}{d^k}$. Since $V_j$ contains $\frac{1}{s}$ elements of $\mathbb{F}^{\times}_q$, we expect the number of solutions $x\in V_j$ to the system of equations $(a_i+x)^{\frac{q-1}{d}}=1 (\forall i\in[k])$ is dominated by  $\frac{q}{sd^k}$, and $k\sqrt{q}$ is the error term.
We postpone the proof of Lemma~\ref{lem:manysolution} to Section~\ref{sec:IntersectionLem}.

\begin{lemma}\label{lem:partitionofneighbor}
For any $r,s,t\in\mathbb{N}^+$ and a prime number $p$ with  $\frac{s}{p-1}+1\le\sqrt{p^{t+1}}$ and $p^{t+1}-1=rs$. Let $g$ be the generator of $\mathbb{F}^{\times}_{p^{t+1}}$. For each $i\in [s]$, let $V_i:=\{g^{i+s},g^{i+2s},\ldots,g^{i+sr}\}$.  Then in the Paley-type bipartite graph $P(p^{t+1},p-1)=(A\;\dot\cup\; B,E)$ with $A=V_1\;\dot\cup\;V_2\;\dot\cup\;\cdots\;\dot\cup\;V_s$, for any $t$ distinct indices $j_1,j_2,\ldots,j_{t}\in[s]$, there exists $\textbf{v}\in V_{j_1}\times V_{j_2}\times\cdots\times V_{j_t}$, such that $|\Gamma(\textbf{v})|\ge p$.
\end{lemma}

\begin{proof}
Fix $t$ distinct indices $j_1,j_2,\ldots,j_{t}\in[s]$.  Consider the sets $S:=V_{j_1}\times V_{j_2}\times\cdots\times V_{j_t}$ and  $\Gamma\langle S\rangle:=\{\{\textbf{v},u\} : \textbf{v}\in S, u\in B, u\in \Gamma(\textbf{v}) \}$.
Since $\frac{s}{p-1}+1\le\sqrt{p^{t+1}} $, for each element $u\in B=\mathbb{F}^\times_{p^{t+1}}$ and $i\in[t]$, applying Lemma~\ref{lem:manysolution} with
\[
q\gets p^{t+1}\quad d\gets p-1\quad k\gets 1\quad a_1\gets u\quad j\gets j_i,
\]
we conclude that the equation $(x+u)^{\frac{p^{t+1}-1}{p-1}}=1$ of $x$ has at least
\[\frac{p^{t+1}}{s(p-1)}- p^{\frac{t+1}{2}}\ge\frac{p^t}{s}+\frac{p^{t-1}}{s}- p^{\frac{t+1}{2}}\ge\frac{p^t}{s}+p^{\frac{t+1}{2}}- p^{\frac{t+1}{2}}=\frac{p^{t}}{s}\] 
solutions in each $V_{j_i}$ ($i\in [t]$). In other words, for all $i\in [t]$, $u$ has
 at least $\frac{p^{t}}{s}$ neighbors in each $V_{j_i}$.
Thus
$|\Gamma\langle S\rangle|\ge(\frac{p^t}{s})^t(p^{t+1}-1)$; on the other hand,  $|S|=(\frac{p^{t+1}-1}{s})^t$. By the pigeonhole principle, there exists $\textbf{v}\in S$ such  that
\begin{equation*}
|\Gamma(\textbf{v})|\ge\frac{|\Gamma\langle S\rangle|}{|S|}\ge\frac{(\frac{p^{t}}{s})^t (p^{t+1}-1)}{(\frac{p^{t+1}-1}{s})^t}
=\frac{p^{t^2}}{(p^{t+1}-1)^{t-1}}\ge\frac{p^{t^2}}{p^{t^2-1}}\ge p.
\end{equation*}
\end{proof}

\subsection{Proof of Lemma~\ref{lem:conthres}}
In the construction of  bipartite graphs satisfying the $(n,k,\ell,h)$-threshold property, we need the  famous Bertrand's Postulate from number theory, whose proof can be found in \cite{ram,erdos}.
\begin{theorem}[Bertrand's postulate ]
For every integer $n>1$, there exists a prime number $p$ such that $n<p<2n$.
\end{theorem}

For any positive integer $n$ and $k=6\ell-1$, by Bertrands's Postulate, we can
choose an arbitrary prime $p$ between $\lceil(n+1)^{\frac{1}{\ell}}\rceil$ and $2\lceil(n+1)^{\frac{1}{\ell}}\rceil$, then we construct the Paley-type graph $P(p^{k+1},p-1)=(A\;\dot\cup\; B,E)$. Let $s:=p^{\ell}-1$, we have $s\ge n$ and $p^{k+1}-1=p^{6\ell}-1=sr$, where $r=(p^{2\ell}+p^{\ell}+1)(p^{3\ell}+1)$. For each $i\in[s]$, let $V_i:=\{g^{i+s},g^{i+2s},\ldots,g^{i+rs}\}$, where $g$ is the generator of $\mathbb{F}^\times_{p^{k+1}}$. 

\medskip

\noindent\textbf{Claim 1}. The graph $P(p^{k+1},p-1)$ including the partition of its vertices set can be constructed in $p^{O(k)}$ time. 
\medskip

\noindent\textit{Proof of Claim 1}. We first find an irreducible polynomial with degree $k+1$ in  $p^{O(k)}$ time. Then we represent elements in $\mathbb{F}_{p^{k+1}}^\times$ as non-zero polynomials with degree at most $k$. To compute the edge set, we go through every pair of elements $(x,y)$  and check if $(x+y)^{\frac{p^{k+1}-1}{p-1}}=1$. 
We find a generator $g$ by enumerating every element in $e\in\mathbb{F}_{p^{k+1}}^\times$ and checking if the set $\{e^i: i\in[p^{k+1}-1]\}$ has $p^{k+1}-1$ elements. 
Then we construct the partition according to the definitions. By Theorem~\ref{thm:A} and Theorem~\ref{thm:B}, these tasks can be done in time $p^{3(k+1)}\cdot (kp)^{O(1)}\le p^{O(k)}$.
\hfill$\dashv$

\medskip

We only need to check $P(p^{k+1},p-1)$ satisfies (T1) and (T2) for parameter $n$, $k$, $\ell\gets (k+1)!$ and $h\gets\lceil(n+1)^{6/(k+1)}\rceil$.

By Theorem~\ref{thm:bicliquefree}, $P(p^{k+1},p-1)$ contains no subgraph isomorphic to $K_{k+1,(k+1)!+1}$, i.e. every $k+1$ distinct vertices in $A$ have at most $(k+1)!$ common neighbors in $B$. Thus $P(p^{k+1},p-1)$ satisfies (T1).

Since $\frac{s}{p-1}+1=\frac{p^\ell-1}{p-1}+1\le p^{3\ell } =\sqrt{p^{{k+1}}}$
, applying Lemma~\ref{lem:partitionofneighbor} with
$t\gets k$, we have for any $k$ distinct indices $a_1,a_2,\ldots,a_{k}\in[s]$, there exist $v_{a_i}\in V_{a_i}$ (for all $i\in[k]$) such  that $v_{a_1},v_{a_2},\ldots,v_{a_{k}}$ have at least $p\ge \lceil(n+1)^{\frac{1}{\ell}}\rceil>(k+1)!$ common neighbors in $B$.

In summary, $P(p^{k+1},p-1)$ with the partition $A=V_1\;\dot\cup \;V_2\;\dot\cup\cdots\dot\cup\;V_s$ satisfies the $(s,k,(k+1)!,\lceil(n+1)^{\frac{1}{\ell}}\rceil)$-threshold property. Note that when $s\ge n$, we can obtain a graph with $(n,k,(k+1)!,\lceil(n+1)^{\frac{1}{\ell}}\rceil)$-threshold property by setting $V_n':=V_n \dot\cup \;V_{n+1}\;\dot\cup\cdots\dot\cup\;V_s$ and returning 
$P(p^{k+1},p-1)$ with the partition $A=V_1\;\dot\cup \;V_2\;\dot\cup\cdots\dot\cup\;V_{n-1}\;\dot\cup \;V_n'$.

\subsection{Proof of the Intersection Lemma}\label{sec:IntersectionLem}
To prove Lemma~\ref{lem:manysolution}, we need to estimate the number of solutions to a system of equations in a subset $V_j$ of a finite field. The main idea is to define a function such that the sum of this function over $V_j$ is related to the number of these solutions.  Weil's theorem can be used as a black box to estimate the sum of this functions over the field. 
We need Lemma~\ref{lem:partition} to lift the sum range to the whole  field.

\begin{definition}[Character]
A character of a finite field $\mathbb{F}_q$ is a function $\chi: \mathbb{F}_q\rightarrow \mathbb{C}$ satisfying the following conditions:
\begin{itemize}
\item[-] $\chi(0)=0$.
\item[-] $\chi(1)=1$.
\item[-] $\forall a,b\in \mathbb{F}_q, \chi(ab)=\chi(a)\chi(b)$.
\end{itemize}
\end{definition}

\begin{remark}
Since for all $x\in \mathbb{F}^\times_q$, $x^{q-1}=1$, we have $\chi(x)^{q-1}=\chi(x^{q-1})=1$. That is $\chi$ maps all the elements in $\mathbb{F}^\times_q$ to the roots of $z^{q-1}=1$ in $\mathbb{C}$.
\end{remark}
\begin{definition}[Order]
A character $\chi$ of a finite field $\mathbb{F}_q$ has \emph{order} $d$ if $d$ is the minimal positive integer such that $\forall a\in \mathbb{F}_q^\times,\,\chi(a)^d=1$.
\end{definition}
\begin{theorem}[A. Weil]
Let $\mathbb{F}_q$ be a finite field, $\chi$  a character of $\mathbb{F}_q$ and $f(x)$  a polynomial over $\mathbb{F}_q$ if:
\begin{itemize}
\item[-] the order of $\chi$ is $d$;
\item[-] for any polynomial $g$ over $\mathbb{F}_q$ and $c\in \mathbb{F}_q$, $f(x)\neq c\cdot (g(x))^d$;
\item[-] the number of distinct roots of $f$ in the algebraic closure of $\mathbb{F}_q$ is $s$,
\end{itemize}
then
$$|\sum_{x\in \mathbb{F}_q}\chi(f(x))|\le(s-1)\sqrt{q}.$$
\end{theorem}(See \cite{Schmidt}, page 43, Theorem 2C')
\begin{remark}
It is well known that the expected translation distance after $n$-step random walk in $2$-dimension space is about $\sqrt{n}$. By the character sum theorem, we can see that the values of $f(x)$ for $x\in \mathbb{F}_q$ distribute randomly to some extent.
\end{remark}

Suppose $g$ is the generator of $\mathbb{F}_q$, where $q$ is a prime power and $q-1=rs$ for some $s,r\in\mathbb{N}$.  For all $i\in[s]$ let $V_i:=\{g^{i+s},g^{i+2s},\ldots,g^{i+rs}\}$. It is obvious that
$\mathbb{F}^\times_q=V_1\cup V_2\cup\cdots\cup V_{s}$ and for all $i\in[s], |V_i|=r$. With these notations, we can show:
\begin{lemma}\label{lem:partition}
Suppose $f$ is a function from $\mathbb{F}_q$ to $\mathbb{C}$, then  for all $i\in[s]$
$$\sum_{z\in V_i} f(z)=\frac{1}{s}\sum_{x\in \mathbb{F}^\times_q} f(g^ix^s).$$
\end{lemma}
\begin{proof}
For any element $z=g^{i+js}\in V_i(j\in[r])$, consider the set
$$X_j:=\{x\in \mathbb{F}^\times_q : g^ix^s=g^{i+js}\}.$$
It is easy to check that $X_j=\{g^{j+r},g^{j+2r},\ldots,g^{j+sr}\}$, i.e. for each $z\in V_i$, the equation  $g^ix^s=z$ of $x$ has exactly $s$ solutions in $\mathbb{F}^\times_q$.
Thus $\sum_{z\in V_i} f(z)=\frac{1}{s}\sum_{x\in \mathbb{F}^\times_q} f(g^ix^s)$.\\
\end{proof}


Now we are ready to prove Lemma~\ref{lem:manysolution}.

\begin{lemma}[Lemma~\ref{lem:manysolution} restated]\label{lem:intersection}
For any $d,k,r,s\in\mathbb{N}^+$ and  prime power $q$ with $q-1=rs$, $d\mid (q-1)$ and $\sqrt{q}\ge\frac{sk}{d}+1$. Let $a_1,a_2,\ldots,a_k$ be $k$ distinct elements in $\mathbb{F}^{\times}_{q}$, $g$ be the generator of $\mathbb{F}^\times_{q}$. For each $j\in [s]$, let $V_j:=\{g^{j+s},g^{j+2s},\ldots,g^{j+sr}\}$. Then for any $j\in [s]$, the number of solutions $x\in V_j$ to the system of equations $(a_i+x)^{\frac{q-1}{d}}=1 (\forall i\in[k])$ is in
$\frac{q}{sd^k}\pm k\sqrt{q}$.
\end{lemma}

\begin{proof}[of Lemma~\ref{lem:manysolution}]
Let $\omega\in\mathbb{C}$ be the primitive $d^{th}$ root of unity, define a function $\chi:\mathbb{F}_q\rightarrow \mathbb{C}$ as follows:
\begin{itemize}
\item[-] $\chi(0):=0$;
\item[-] for $g^\ell\in \mathbb{F}^\times_q$ set $\chi(g^\ell):=\omega^\ell$.
\end{itemize}
Then:
\begin{description}
\item[i] $\chi$ is  a character of $\mathbb{F}_q$. Because $\chi(g^a\cdot g^b)=\omega^{a+b}=\chi(g^a)\chi(g^b)$ and 
$\chi(1)=\chi(g^{q-1})=w^{q-1}=1$ by $d\mid q-1$.
\item[ii] The order of $\chi$ is $d$. Observe that for $n\in\mathbb{N}^+$: $\chi(g)^n=\chi(g^n)=1\iff \omega^n=1\iff d\mid n$, thus the order of $\chi$ is $\ge d$. On the other hand, note that $g$ is a generator of $\mathbb{F}_q^\times$. For all
$z\in\mathbb{F}_q^\times$, there exists $i_z\in [q-1]$ such that $z=g^{i_z}$. Thus $\chi(z)^d=\chi(g^{i_zd})=\omega^{di_z}=1$, which implies that the order of $\chi$ is $\le d$.
\item[iii] $\chi(x)=1\iff x^{\frac{q-1}{d}}=1$. Suppose $x=g^i$ and notice that $g^\ell=1 \iff q-1\mid\ell$, it follows that $ 1=x^\frac{q-1}{d}=g^{\frac{i(q-1)}{d}}\iff q-1\mid \frac{i(q-1)}{d}\iff d\mid i\iff \omega^i=1\iff\chi(x)=\chi(g^i)=1$.
\end{description}


By iii, $(a_i+x)^{\frac{q-1}{d}}=1\iff \chi(a_i+x)=1$,  let $$X:=\{x\in V_j : \forall i\in[k],\chi(x+a_i)=1 \}.$$
Recall that $a\pm b$ denotes the set of real number between $a-b$ and $a+b$, our goal is to show that $|X|\in\frac{q}{sd^k}\pm k\sqrt{q}$.

Define a polynomial $h\colon\mathbb{C}\to\mathbb{C}$ by setting $h(z):=\frac{z^{d}-1}{z-1}=1+z+\ldots+z^{d-1}$, then:
\begin{description}
\item $h(1)=d$;
\item $h(\omega^i)=0$, for all $i\in [d-1]$;
\item $h(0)=1$.
\end{description}

Let $H(x):=\prod_{i=1}^k h(\chi(a_i+x))$. It is easy to check that:
\begin{description}
\item if $x\in X$, then $H(x)=d^k$;
\item if $x=-a_i$ for some $i\in[k]$ and $\chi(x+a_{i'})=1 (\forall i'\in[k],i'\neq i)$, then $H(x)=d^{k-1}$;
\item otherwise $H(x)=0$.
\end{description}

Now consider the sum $S:=\sum_{x\in V_j}H(x)$, we have
\begin{equation}\label{eq:XandS}
|X|d^k\le S\le |X|d^k+kd^{k-1}.
\end{equation}
We only need to estimate $S$. Using Lemma \ref{lem:partition}, we can rewrite $S$ as
\begin{align*}
S &= \sum_{x\in V_j}H(x)\\
&= \frac{1}{s}\sum_{x\in \mathbb{F}^\times_q}H(g^jx^s)\\
&= \frac{1}{s}[\sum_{x\in \mathbb{F}_q}H(g^jx^s)-H(0)].
\end{align*}
Expand the products in $H(g^jx^s)$, we get
\begin{align*}
&\sum_{x\in \mathbb{F}_q}H(g^jx^s) \\
=& \sum_{x\in \mathbb{F}_q}\prod_{i=1}^k h(\chi(a_i+g^jx^s))\\
=& \sum_{x\in \mathbb{F}_q}\prod_{i=1}^k [1+\chi(a_i+g^jx^s)+\ldots+\chi(a_i+g^jx^s)^{d-1}]\\
=& \sum_{x\in \mathbb{F}_q}\sum_{\psi\in \{0,1,\ldots,d-1\}^k}\chi(f_\psi(x))\\
=& q+\sum_{\psi\in \{0,1,\ldots,d-1\}^k\setminus\{0\}^k}\sum_{x\in \mathbb{F}_q}\chi(f_\psi(x)),
\end{align*}
where $\psi\in\{0,1,\ldots,d-1\}^k$ is a function from $[k]$ to $\{0,1,\ldots,d-1\}$  and $f_\psi(x):=\prod_{i=1}^k(a_i+g^jx^s)^{\psi(i)}$.

\medskip

To invoke Weil's theorem on the character sum $\sum\chi(f_\psi(x))$ for every $\psi\in\{0,1,\ldots,d-1\}^k\setminus\{0\}^k$, we need to check:
\begin{itemize}
\item[-] The order of $\chi$ is $d$, this is done in the previous discussion.

\item[-]  $f_\psi(x)\neq c\cdot(g(x))^d$ for any polynomial $g$ over $\mathbb{F}_q$ and $c\in \mathbb{F}_q$. It suffices to show that any solution of $f_\psi(x)$ in the algebraic closure of $\mathbb{F}_q$ has  multiplicity $\le d-1$. For each $i\in[k]$, let $f_{i}(x)=a_i+g^jx^s$. Note that the derivative of $f_{i}(x)$ is $f_{i}'(x)=s\cdot g^jx^{s-1}$. We claim that all the roots of $f_{i}(x)$ have multiplicity $1$, otherwise by Theorem~\ref{thm:multiplicity} $f_{i}(x)$ and $f_{i}'(x)$ have a common root $\alpha$.  We must have
$s\cdot a_i=s\cdot(a_i+g^j\alpha^s)-(s\cdot g^j\alpha^s)=s\cdot f_i(\alpha)-f'_i(\alpha)\alpha=0$.
 This is impossible because $q-1=sr$ and Lemma~\ref{lem:orderp} implies $(rs)\cdot a_i=-a_i\neq 0$; on the other hand, for any distinct $i,i'\in[k]$, $f_{i}(x)$ and $f_{i'}(x)$ do not share a common root because $a_i\neq a_{i'}$. Since $f_\psi=\prod_{i=1}^kf_i(x)^{\psi(i)}$,  each root of $f_\psi$ has multiplicity $\le d-1$.
\item[-] $f_\psi$ has at most $ks$ distinct roots in the algebraic closure field of $\mathbb{F}_q$. This follows by the facts that $f_\psi$ is a product of $k$ polynomials and each polynomial has at most $s$ distinct roots. 
\end{itemize}
By Weil's theorem
$$|\sum_{x\in \mathbb{F}_q}\chi(f_\psi(x))| \le (ks-1)\sqrt{q}.$$
So
\begin{align*}
|S+\frac{H(0)}{s}-\frac{q}{s}|&=\frac{1}{s}\sum_{\psi\in \{0,1,\ldots,d-1\}^k\setminus\{0\}^k}\sum_{x\in \mathbb{F}_q}\chi(f_\psi(x))\\
&\le \frac{d^k}{s}(ks-1)\sqrt{q}.
\end{align*}
Thus we obtain the following estimate for $S$:
\begin{equation}\label{eq:estimateS}
|S|\in \frac{q-H(0)}{s}\pm \frac{d^k}{s}(ks-1)\sqrt{q}.
\end{equation}
Finally, notice that $H(0)\le d^k$ and $\sqrt{q}>\frac{sk}{d}+1$, we have
\begin{align*}
|X|&\in\frac{S}{d^k}\pm \frac{k}{d} \quad\quad\text{(by  (\ref{eq:XandS}))}\\
&\subseteq\frac{q-H(0)\pm(ks-1)d^k\sqrt{q}}{sd^k}\pm \frac{k}{d} \quad\quad\text{ (by  (\ref{eq:estimateS}))}\\
&\subseteq\frac{q}{sd^k}\pm (k\sqrt{q}+\frac{k}{d}+\frac{1}{s}-\frac{\sqrt{q}}{s}) \quad\quad\text{ (using  $H(0)\le d^k$)}\\
&\subseteq\frac{q}{sd^k}\pm k\sqrt{q}.  \quad\quad\text{ (using  $\sqrt{q}>\frac{sk}{d}+1$)}
\end{align*}
\end{proof}

\section{Probabilistic construction}\label{sec:probaconst}
One important methodology we learn from  Erd\H{o}s's paper~\cite{erdos1959graph} on graph theory and probability  is that to prove  some graphs with a certain property exist, it suffices to demonstrate that the probability of such graphs is positive in some probability space. 

For $n\in\mathbb{N}^+$ and $p : \mathbb{N}\to [0,1]$. We define  $\mathcal{G}(n,p)=(\Omega,\Pr)$ as the probability space of \emph{bipartite random graphs} where $\Omega$ is the set of all bipartite graphs $G=(A\;\dot\cup\; B,E)$  with $A=B=[n^2]$  and  for each  graph $G\in\Omega$, $\Pr(G):=p(n)^{|E(G)|}(1-p(n))^{{n}^{4}-|E(G)|}$.  To shorten notation, we write $p$ instead of $p(n)$. We use $G(n,p)$ to denote the random graph sampled from 
$\mathcal{G}(n,p)$.
The probability of $G(n,p)$ satisfying some property $P$ is the probability of the event $E:=\{G\in\Omega : \text{$G$ satisfies $P$}\}$.

We  partition $A$ into $n$ subsets $(V_1,V_2,\ldots, V_{n})$ with $V_i:=\{1+(i-1)n,2+(i-1)n,\ldots,n+(i-1)n\}$ for each $i\in[{n}]$. We will show that with high probability $G(n,p)$ with the partition $A=V_1\;\dot\cup\;V_2\;\dot\cup\cdots\dot\cup\; V_{n}$ satisfies the  $(n,k,\ell,h)$-threshold property for  $\ell=2k^2+4k-1<h\le n^{\frac{1}{4(k+1)}}$ and $p=n^{-\frac{2(k+\ell+3)}{(k+1)(\ell+1)}}$.  To that end, we
use Lemma~\ref{lem:T1} to give an upper bound for $\Pr(\text{$G(n,p)$ does not satisfy (T1)})$ and
 Lemma~\ref{lem:T2} to give an upper bound for  $\Pr(\text{$G(n,p)$ does not satisfy (T2)})$.

\medskip

Before giving these upper bounds, we need some preparations.  Observe that the vertex set of graphs in $\Omega$ is fixed, we can identify every graph in $\Omega$ with its edge set.  Let $E_{A,B}$ be the set of edges between $A$ and $B$. For every edge set $E\subseteq E_{A,B}$, let $G_E:=(A\;\dot\cup\;B,E)$. 
We say an event $X$ depends only on a set $E\subseteq E_{A,B}$ if there exists a  $C\subseteq 2^E$ such that
\[
X=\{G\in\Omega : E(G)\cap E\in C\}.
\]
We call $C$ the restriction of $X$ to $E$.
\begin{lemma}\label{lem:subspace}
Suppose  $E_1\subseteq E_{A,B}$ and $X$ is an event depending only on $E_1$. Let $C\subseteq 2^{E_1}$ be the restriction of $X$ to $E_1$.
Then we have
\[
\Pr(X)=\sum_{E'\in C}p^{|E'|}(1-p)^{|E_1|-|E'|}.
\]
\end{lemma}
\begin{proof}
Let $E_2:=E_{A,B}\setminus E_1$. Note that $n^4=|E_{A,B}|=|E_1|+|E_2|$.
\begin{align*}
\Pr(X)&=\sum_{G\in X}\Pr(G)\\
&=\sum_{G\in \Omega,E(G)\cap E_1\in C}\Pr(G)\\
&=\sum_{E'\in C}\sum_{E''\in 2^{E_2}}\Pr(G_{E'\cup E''})\\
&=\sum_{E'\in C}\sum_{E''\in 2^{E_2}}p^{|E'|+|E''|}(1-p)^{n^4-|E'|-|E''|}\\
&=\sum_{E'\in C}p^{|E'|}(1-p)^{|E_1|-|E'|}\sum_{E''\in 2^{E_2}}p^{|E''|}(1-p)^{|E_2|-|E''|}\\
&=\sum_{E'\in C}p^{|E'|}(1-p)^{|E_1|-|E'|}.
\end{align*}
\end{proof}

\begin{lemma}\label{lem:independentVar}
Suppose $E_{A,B}=E_1\;\dot\cup\; E_2$. 
Event $X$ depends only on $E_1$. Event $Y$ depends only on $E_2$. 
Then we have
\[
\Pr(X\cap Y)=\Pr(X)\Pr(Y).
\]
\end{lemma}
\begin{proof}
Let $C$ be the restriction of $X$ to $E_1$, $D$ be the restriction of $Y$ to $E_2$.
By Lemma~\ref{lem:subspace},
$\Pr(X)=\sum_{E'\in C}p^{|E'|}(1-p)^{|E_1|-|E'|}$ and 
$\Pr(Y)=\sum_{E'\in D}p^{|E'|}(1-p)^{|E_2|-|E'|}$.
\begin{align*}
\Pr(X\cap Y)&=\sum_{G\in X\cap Y}\Pr(G)\\
&=\sum_{E'\in C}\sum_{E''\in D}\Pr(G_{E'\cup E''})\\
&=\sum_{E'\in C}\sum_{E''\in D}p^{|E'|+|E''|}(1-p)^{n^4-|E'|-|E''|}\\
&=\sum_{E'\in C}p^{|E'|}(1-p)^{|E_1|-|E'|}\sum_{E''\in D}p^{|E''|}(1-p)^{|E_2|-|E''|}\\
&=\Pr(X)\Pr(Y).
\end{align*}
\end{proof}

For $S\subseteq A$ and $T\subseteq B$, we define a random variable $X_{S,T} : \Omega\to \mathbb{R}$ as follows. For every $G\in\Omega$, let $X_{S,T}(G):=1$ if $T\cup S$ forms a complete bipartite subgraphs in $G$, otherwise $X_{S,T}(G)=0$. Let $E_1:= \{\{s,t\}:s\in S,t\in T\}$, $C:= \{E_1\}$ and $X:=\{G\in\Omega: X_{S,T}(G)=1\}$. It is straightforward to verify that $X$ depends only on $E_1$ and $C$ is the restriction of $X$ to $E_1$.
Applying Lemma~\ref{lem:subspace}, we get $\Pr(X_{S,T}=1)=\Pr(X)=p^{|S|\cdot |T|}$. Thus
\begin{equation}\label{eq:ExpectST}
\mathrm{E}[X_{S,T}]=\Pr(X_{S,T}=1)=p^{|S|\cdot |T|}.
\end{equation}
\subsection{Estimate for $\Pr(\text{$G(n,p)$ does not satisfy (T1)})$}

To bound the probability of $G(n,p)$ containing a subgraph isomorphic to $K_{k+1,h}$, we need the following lemma, which is  a simple consequence of Markov's Inequality.
\begin{lemma}
Let $X$ be a nonnegative integral  random variable, then
$\Pr(X>0)\le \mathrm{E}[X]$.
\end{lemma}

\begin{lemma}\label{lem:T1}
With probability at most $n^{-2}$, $G(n,p)$ does not satisfy (T1).

\end{lemma}
\begin{proof}
Let $X$ be the number of $K_{k+1,\ell+1}$-subgraphs in $G(n,p)$ with the left $k+1$ vertices in $A$ and the other $\ell+1$ vertices in $B$, i.e., for every $G\in\Omega$
\[
X(G):=\left|\left\{(A',B') :  A'\in \binom{A}{k+1}, B'\in \binom{B}{\ell+1}, B'\subseteq\Gamma^G(A')\right\}\right|.
\]

 Then by the linearity of expectation
\begin{align*}
\mathrm{E}[X]&=\sum_{A'\in\binom{A}{k+1},B'\in\binom{B}{\ell+1}}\mathrm{E}[X_{A',B'}]\\
&=\binom{n^2}{k+1}\cdot\binom{n^2}{\ell+1}\cdot p^{(k+1)(\ell+1)} \quad\text{(using (\ref{eq:ExpectST}))}\\
&\le n^{2(k+1+\ell+1)}\cdot n^{-2(k+\ell+3)}\\
&= n^{-2}.
\end{align*}
We have $\Pr(X>0)\le \mathrm{E}[X]\le n^{-2}$.
It follows from the definition that
\[
\Pr(\text{$G(n,p)$ does not satisfy (T1)})\le\Pr(X>0)\le n^{-2}.
\]
\end{proof}
Hence, when  $n\rightarrow\infty$, $ G(n,p)$ satisfies the first condition of $(n,k,\ell,h)$-threshold property with high probability.

\subsection{Estimate for $\Pr(\text{$G(n,p)$ does not satisfy (T2)})$}

For $S\in\binom{A}{k}$ define a random variable $X_S : \Omega\to\mathbb{R}$ such that for every $G\in\Omega$,  
\begin{equation*}
X_S(G):=\left|\left\{T\in\binom{B}{h} : T\subseteq\Gamma^G(S)\right\}\right|.
\end{equation*}
In other words, $X_S$ denotes the number of $K_{k,h}$-subgraphs in $G(n,p)$ whose left side vertex set is $S$. 
\begin{lemma}\label{lem:XS}
 If $h\le n^{\frac{1}{4(k+1)}}$, then $\Pr(X_S=0)\le n^{-\frac{1}{4(k+1)}}$.
\end{lemma}
\begin{proof}
By the Chebyshev's Inequality, $\Pr(X_S=0)\le\frac{\mathrm{Var}[X_S]}{\mathrm{E}[X_S]^2}$. To bound $\Pr(X_S=0)$, we need to estimate $\mathrm{E}[X_S]$ and $\mathrm{Var}[X_S]$. By the linearity of expectation and (\ref{eq:ExpectST}), we have
\begin{equation}\label{eq:ES}
\mathrm{E}[X_S]=\sum_{T\in\binom{B}{h}}\mathrm{E}[X_{S,T}]=\binom{n^2}{h}\cdot p^{kh}.
\end{equation}
It follows that
 \begin{align*}
 & \mathrm{Var}[X_S]\\
 =&\mathrm{E}[X_S^2]-\mathrm{E}[X_S]^2\\
 =&\mathrm{E}[(\sum_{T\in\binom{B}{h}}X_{S,T})^2]-\mathrm{E}[X_S]^2\\
 =&\sum_{T,T'\in\binom{B}{h}}\mathrm{E}[X_{S,T}X_{S,T'}]-\mathrm{E}[X_S]^2\\
 =&\sum_{i=0}^{h}\sum_{T,T'\in\binom{B}{h},|T\cap T'|=i}\mathrm{E}[X_{S,T\cup T'}]-\mathrm{E}[X_S]^2\\
 =&\sum_{i=0}^{h}\sum_{T,T'\in\binom{B}{h},|T\cap T'|=i}p^{|S|\cdot |T\cup T'|}-\mathrm{E}[X_S]^2\quad \text{(using (\ref{eq:ExpectST}))}\\
 =&\sum_{i=0}^{h}\sum_{T,T'\in\binom{B}{h},|T\cap T'|=i} p^{2hk-ik}-\mathrm{E}[X_S]^2 \\
  =&\sum_{i=0}^h \binom{n^2}{h}\binom{n^2-h}{h-i}\binom{h}{i}\cdot p^{2hk-ik}-\mathrm{E}[X_S]^2\\
\le&\sum_{i=1}^h \binom{n^2}{h}\binom{n^2-h}{h-i}\binom{h}{i}\cdot p^{2hk-ik}\quad \text{(using $\binom{n^2-h}{h}\le\binom{n^2}{h}$ and (\ref{eq:ES}))}\\
 = &\binom{n^2}{h}\binom{n^2}{h}p^{2hk}\sum_{i=1}^h \frac{\binom{n^2-h}{h-i}\binom{h}{i}p^{-ik}}{\binom{n^2}{h}}\\
 = &\mathrm{E}[X_S]^2\sum_{i=1}^h \frac{\binom{n^2-h}{h-i}\binom{h}{i}p^{-ik}}{\binom{n^2}{h}} \quad \text{(using (\ref{eq:ES}))}\\ 
  \le &\mathrm{E}[X_S]^2\sum_{i=1}^h h^{2i}{n}^{-2i}p^{-ik} 
\quad\text{(using $\binom{n^2-h}{h-i}\binom{n^2}{i}\le\binom{n^2}{h}\binom{h}{i}$ and $\binom{h}{i}/h^i\le\binom{n^2}{i}/n^{2i}$)}\footnotemark\\ 
\le &\mathrm{E}[X_S]^2 \sum_{i=1}^h n^{-2i[1-\frac{k(k+\ell+3)}{(k+1)(\ell+1)}-\frac{1}{4(k+1)}]}\quad\text{(using $h\le n^{\frac{1}{4(k+1)}}$ and $p=n^{-\frac{2(k+\ell+3)}{(k+1)(\ell+1)}}$)}\\
    = &\mathrm{E}[X_S]^2 \sum_{i=1}^h n^{\frac{-i}{2(k+1)}}\quad\text{(using $\ell=2k^2+4k-1$)}\\ 
 \le &\mathrm{E}[X_S]^2\cdot  hn^{-\frac{1}{2(k+1)}}\\
\le &\mathrm{E}[X_S]^2\cdot  n^{-\frac{1}{4(k+1)}} \quad\text{(using $h\le n^{\frac{1}{4(k+1)}}$)}.
 \end{align*}
 \footnotetext{Observe that  $\binom{n^2-h}{h-i}\binom{n^2}{i}\le\binom{n^2-i}{h-i}\binom{n^2}{i}=\binom{n^2}{h}\binom{h}{i}$ and the function $f(a)=\binom{a}{i}/a^i$ is monotone.}
Applying the Chebyshev's Inequality, we obtain $\Pr(X_S=0)\le n^{-\frac{1}{4(k+1)}}$.
\end{proof}

\begin{lemma}\label{lem:T2}
If $n\ge 4(1+k)^2$ and $h\le n^{\frac{1}{4(k+1)}}$, then
with probability at most $n^{-1}$ $G(n,p)$ does not satisfy the second condition of $(n,k,\ell,h)$-threshold property.
\end{lemma}
\begin{proof}
For $I=\{a_1,a_2,\ldots,a_k\}\in\binom{[n]}{k}$, consider the subsets $V_{a_1},V_{a_2},\ldots,V_{a_k}$ in the partition of $A$. Recall that $V_{a_i}=\{1+(a_i-1)n,2+(a_i-1)n,\ldots,n+(a_i-1)n\}$  for each $i\in[k]$. Denote by $Y_I$  the number of $K_{k,h}$-subgraph in $G(n,p)$ with the  restriction that each $V_{a_i} (i\in [k])$ contains exactly one vertex from  the left side of such $K_{k,h}$-subgraphs. For each $j\in[n]$ let $S_j:=\{j+(a_1-1)n,j+(a_2-1)n,\ldots,j+(a_k-1)n\}$. For  each $G\in\Omega$ let  $X_{S_j}(G)$ be the number of $h$-vertex sets in $\Gamma^G(S_j)$. 
By Lemma~\ref{lem:XS}, $\Pr(X_{S_j}=0)\le n^{-\frac{1}{2(k+1)}}$. 

\medskip
\noindent \textit{Claim 1}.  $\Pr(\forall j\in[{n}], X_{S_j}=0)=\prod_{j=1}^{{n}}\Pr(X_{S_j}=0)$.

\medskip

\noindent\textit{Proof of the Claim 1}. Let 
\[
X:=\{G:X_{S_n}(G)=0\}
\] 
and 
\[
Y:=\{G:\text{for all $j\in[n-1]$, $X_{S_j}(G)=0$}\}.
\]
It suffices to show that $\Pr(X\cap Y)=\Pr(X)\Pr(Y)$.
Let $E_1:=\{\{a,b\}:a\in S_n,b\in B\}$ and $E_2:=E_{A,B}\setminus E_1$. 
 Note that for all $j\in[n-1]$, $S_j\cap S_{n}=\varnothing$. By  definitions, we have that 
$X$ depends only on $E_1$ and $Y$ depends only on   $E_2$.
Applying Lemma~\ref{lem:independentVar}, we get $\Pr(X\cap Y)=\Pr(X)\Pr(Y)$.
\hfill$\dashv$

\medskip

Putting all together, we obtain
\[
\Pr(Y_I=0)\le \Pr(\forall j\in[{n}], X_{S_j}=0)=\prod_{j=1}^{{n}}\Pr(X_{S_j}=0)\le n^{-\frac{n}{4(k+1)}}.
\]
$G(n,p)$ does not satisfy the second condition of threshold property if there exists $I\in\binom{[n]}{k}$ such that $Y_I=0$. By the union bound
\[
\Pr(G(n,p)\mbox{ does not satisfy (T2)})\le \sum_{I\in\binom{[n]}{k}}\Pr(Y_I=0)\le n^{k-\frac{{n}}{4(k+1)}}\le n^{-1}.
\]
\end{proof}

\subsection{Proof of Lemma \ref{lem:conprob}}
Choose $n$ large enough such that ${n}\ge 4(1+k)^2$ and $n\ge 20$, then
from Lemma \ref{lem:T1} and Lemma \ref{lem:T2} we can deduce
\[
\Pr(G(n,p)\mbox{ does not satisfy T1 or T2})\le  n^{-2}+n^{-1}\le 1/10.
\]

Thus $G(n,p)$ satisfies the $({n},k,\ell,h)$-threshold property  with  probability larger than $9/10$.

\section{Conclusions}\label{sec:conclusion}
The main result of this paper is to give an \fpt-reduction from $\kclique$ to $\kbiclique$, thus proving the $\Wone$-hardness of $\kbiclique$. Our reduction for $\kbiclique$ uses a class of graphs satisfying the $(n,k,\ell,h)$-threshold property, which might be of some independent interest. We give a probabilistic construction of graphs with $(n,k,\ell,h)$-threshold property for $\ell=\Theta(k^2)$, which yields an $f(k)\cdot n^{o(\sqrt{k})}$-time lower bound for $\kbiclique$ for any computable function $f$ assuming a randomized version of $\ETH$. An immediate open question is to give an explicit construction of graphs satisfying the $(n,k,\ell,h)$-threshold property for $\ell=\Theta(k^2)$, hence obtain the same lower bound for $\kbiclique$ under $\ETH$. Another obvious question for further research is whether there exists $f(k)\cdot n^{o(k)}$-time algorithm solving $\kbiclique$. We believe that the answer is negative. Note that $\kbiclique$ is a special case of \psubiso. The dichotomy classification of \psubiso\ is still unclear. It remains to be seen if we can prove the $\Wone$-hardness of \psubiso\ for other graph classes $\mathbf{C}$ with unbounded tree-width using $\kbiclique$ as a reduction source.

It is worth pointing out that our reduction creates a gap on one side of the biclique. Such a gap can be used to prove parameterized inapproximability results of other optimization problems~\cite{Chen2016TheCI} (for the definition of parameterized approximability, see ~\cite{cai2006fixed,downey2006parameterized,chen2006parameterized,marx2008parameterized}).

On the algorithmic side, it is of interest to investigate the trade-off between the running-time of  algorithms and the approximation ratios for \textsc{Maximum-$k$-Subset-Intersection}. More precisely, for any $t : \mathbb{N}\to\mathbb{N}$ and $r : \mathbb{N}\to\mathbb{N}$, we want to know if there exist algorithms that approximate \textsc{Maximum-$k$-Subset-Intersection}  to ratio $n^{r(k)}$ in  $f(k)\cdot n^{t(k)}$ time for some computable function $f$. Result of this paper rules out the existence of algorithms for $t(k)=o(\sqrt{k})$ and $r(k)=o(1/\sqrt{k})$ under the Exponential Time Hypothesis.




\noindent\textbf{Acknowledgement}
The author would like to thank  Yijia Chen, Hiroshi Imai and the anonymous reviewers  for their valuable comments and suggestions to improve the paper.


\bibliographystyle{alpha} 
\bibliography{ref}

\end{document}

%% file: reductionfig.tex
\begin{tikzpicture}[scale=0.94]

\tikzstyle{mynodestyle} = [scale=0.3,draw=black,color=black,circle]

\node at (-5,-2.5) {$G$};

\node (v29)[mynodestyle] at (-5,4) {};
\node (v30)[mynodestyle] at (-5.5,3.5) {};
\node (v31)[mynodestyle]  at (-5,3) {};
\node  at (-5.2,4) {$a$};
\node  at (-5.7,3.5) {$b$};
\node  at (-4.8,3) {$c$};

\node (v32)[mynodestyle] at (-5,0.5) {};
\node (v33)[mynodestyle]  at (-5.5,0) {};
\node (v35)[mynodestyle]  at (-5,-0.5) {};
\node (v34)[mynodestyle]  at (-4.5,0) {};
\node  at (-5.2,0.5) {$a$};
\node  at (-5.7,0) {$b$};
\node  at (-4.8,-0.5) {$c$};
\node  at (-4.3,0) {$d$};

\draw [thick] (v29) edge (v30);
\draw [thick] (v30) edge (v31);
\draw [thick] (v29) edge (v31);
\draw [thick] (v32) edge (v33);
\draw [thick] (v34) edge (v35);
\draw [thick] (v32) edge (v35);

\node at (-0.5,-2.5) {$T$};
\node at (-1.7,-2) {$A$};

\node at (0.7,-2) {$B$};
\node at (-6.5,3.5) {(yes)};
\node at (-6.5,0) {(no)};
\draw[thick]  (-5,3.5) ellipse (0.8 and 1.5);
\draw[thick]  (-5,0) ellipse (0.8 and 1.5);

\draw[thick]  (-2.5,5) rectangle (-1,2);
\draw[thick]  (0,5) rectangle (1.5,2);

\draw[thick]  (-2.5,1.5) rectangle (-1,-1.5);
\draw[thick]  (0,1.5) rectangle (1.5,-1.5);

\node (v5) at (-2.5,4.5) {};
\node (v6) at (-1,4.5) {};
\node (v3) at (-2.5,4) {};
\node (v4) at (-1,4) {};

\node (v1) at (-2.5,3.5) {};
\node (v2) at (-1,3.5) {};
\draw [dashed] (v1) edge (v2);
\draw [dashed] (v3) edge (v4);
\draw [dashed] (v5) edge (v6);

\draw[dotted]  (-1.9,3.7) ellipse (0.4 and 0.8);

\node (va) at (-1.9,4.28) [mynodestyle] {};
\node (vva) at (-2.2,4.28)  {$u_a$};
\node (vb) at (-1.7,3.78) [mynodestyle] {};
\node (vvb) at (-2,3.78)  {$u_b$};
\node (vc) at (-1.9,3.28) [mynodestyle] {};
\node (vvc) at (-2.2,3.28)  {$u_c$};

\node (v8)  at (-2,4.5) {};
\node (v9) at (-2,2.9) {};
\node (v11) at (-2.5,3) {};
\node (v10) at (-1,3) {};
\node (v12) at (-2.5,2.5) {};
\node (v13) at (-1,2.5) {};
\draw[dashed]  (v10) edge (v11);
\draw[dashed]  (v12) edge (v13);
\draw[ dotted]  (0.8,3.8) ellipse (0.5 and 1);
\node (v7) at (0.8,4.8) {};
\node (v14) at (0.8,2.8) {};
\draw[dashdotted]  (v8) edge (v7);
\draw[dashdotted]  (v9) edge (v14);

\node at (0.8,3.8) {$h$};
\node (v15) at (-2.5,1) {};
\node (v16) at (-1,1) {};
\node (v17) at (-2.5,0.5) {};
\node (v18) at (-1,0.5) {};
\node (v19) at (-2.5,0) {};
\node (v20) at (-1,0) {};
\node (v21) at (-2.5,-0.5) {};
\node (v22) at (-1,-0.5) {};
\node (v23) at (-2.5,-1) {};
\node (v24) at (-1,-1) {};

\draw[dashed]  (v15) edge (v16);
\draw[dashed]  (v17) edge (v18);
\draw [dashed] (v19) edge (v20);
\draw [dashed] (v21) edge (v22);
\draw [dashed] (v23) edge (v24);

\node (dva) at (-1.9,0.78) [mynodestyle] {};
\node (dvva) at (-2.2,0.78)  {$u_a$};
\node (dvb) at (-1.7,0.28) [mynodestyle] {};
\node (dvvb) at (-2,0.28)  {$u_b$};
\node (dvc) at (-1.9,-0.22) [mynodestyle] {};
\node (dvvc) at (-2.2,-0.22)  {$u_c$};
\node (dvd) at (-1.9,-0.72) [mynodestyle] {};
\node (dvvd) at (-2.2,-0.72)  {$u_d$};

\draw[ dotted]  (-1.9,0.08) ellipse (0.4 and 0.9);
\draw[ dotted]  (0.8,0.3) ellipse (0.2 and 0.4);

\node (v25) at (-2,1) {};
\node (v27) at (-2,-0.82) {};
\node (v26) at (0.8,0.7) {};
\node (v28) at (0.8,-0.1) {};
\draw [dashdotted] (v25) edge (v26);
\draw [dashdotted] (v27) edge (v28);
\node at (0.8,0.3) {$l$};
\node at (-3,4.3) {$V_a$};
\node at (-3,3.8) {$V_b$};
\node at (-3,3.3) {$V_c$};

\node at (-3,0.8) {$V_a$};
\node at (-3,0.3) {$V_b$};
\node at (-3,-0.2) {$V_c$};
\node at (-3,-0.7) {$V_d$};

\draw [thick] (6,5) rectangle (7.5,2);
\draw [thick,rounded corners] (2.8,-1.5) rectangle (5 , 1.5);
\draw [thick] (6,1.5) rectangle (7.5,-1.5);
\draw [thick,rounded corners] (2.8,2) rectangle (5 , 5);
\draw[ dotted]  (6.8,3.8) ellipse (0.5 and 1);
\node (v7) at (6.8,4.8) {};
\node (v14) at (6.8,2.8) {};
\draw[ dotted]  (6.8,0.3) ellipse (0.2 and 0.4);
\draw [ dotted] (4,3.5) ellipse (0.4 and 0.7);
\node [mynodestyle] at (4,4) {};
\node [mynodestyle] at (4,3.5) {};
\node [mynodestyle] at (4,3) {};
\node  at (3.5,4) {$u_au_b$};
\node  at (3.5,3.5) {$u_au_c$};
\node  at (3.5,3) {$u_bu_c$};

\draw [ dotted] (4,0) ellipse (0.4 and 0.7);
\node [mynodestyle] at (4,0.5) {};
\node [mynodestyle] at (4,0) {};
\node [mynodestyle] at (4,-0.5) {};
\node  at (3.5,0.5) {$u_au_b$};
\node  at (3.5,0) {$u_au_c$};
\node  at (3.5,-0.5) {$u_cu_d$};
\node (v36) at (4,4.2) {};
\node (v37) at (4,2.8) {};

\draw [dashdotted] (v36) edge (v7);
\draw [dashdotted] (v37) edge (v14);
\node at (6.8,3.8) {$h$};
\node (v38) at (4,0.7) {};
\node (v40) at (4,-0.7) {};
\node (v39) at (6.8,0.7) {};
\node (v41) at (6.8,-0.1) {};
\draw [dashdotted] (v38) edge (v39);
\draw [dashdotted] (v40) edge (v41);
\node at (6.8,0.3) {$l$};
\node at (5.5,-2.5) {$H$};
\node at (4,-2) {$L$};
\node at (6.75,-2) {$R$};
\end{tikzpicture}

%% file: bicliquearxiv.bbl
\newcommand{\etalchar}[1]{$^{#1}$}
\begin{thebibliography}{BRFGL10}

\bibitem[AL86]{adleman1986finding}
Leonard~M Adleman and Hendrik~W Lenstra.
\newblock Finding irreducible polynomials over finite fields.
\newblock In {\em Proceedings of the eighteenth annual ACM symposium on Theory
  of computing}, pages 350--355. ACM, 1986.

\bibitem[ALR12]{AtminasLR12}
Aistis Atminas, Vadim~V. Lozin, and Igor Razgon.
\newblock Linear time algorithm for computing a small biclique in graphs
  without long induced paths.
\newblock In {\em Scandinavian Workshop on Algorithm Theory}, pages 142--152.
  Springer, 2012.

\bibitem[AMS11]{ambuhl2011inapproximability}
Christoph Amb{\"u}hl, Monaldo Mastrolilli, and Ola Svensson.
\newblock Inapproximability results for maximum edge biclique, minimum linear
  arrangement, and sparsest cut.
\newblock {\em SIAM Journal on Computing}, 40(2):567--596, 2011.

\bibitem[Aro98]{arora1998approximability}
Sanjeev Arora.
\newblock The approximability of np-hard problems.
\newblock In {\em Proceedings of the thirtieth annual ACM symposium on Theory
  of computing}, pages 337--348. ACM, 1998.

\bibitem[AYZ95]{AlonYZ95}
Noga Alon, Raphael Yuster, and Uri Zwick.
\newblock Color-coding.
\newblock {\em J. ACM}, 42(4):844--856, 1995.

\bibitem[BGK{\etalchar{+}}96]{BabaiGKRSW96}
L{\'a}szl{\'o} Babai, Anna G{\'a}l, J{\'a}nos Koll{\'a}r, Lajos R{\'o}nyai,
  Tibor Szab{\'o}, and Avi Wigderson.
\newblock Extremal bipartite graphs and superpolynomial lower bounds for
  monotone span programs.
\newblock In {\em Proceedings of the twenty-eighth annual ACM symposium on
  Theory of computing}, pages 603--611. ACM, 1996.

\bibitem[BM14]{BulatovM11}
Andrei~A Bulatov and D{\'a}niel Marx.
\newblock Constraint satisfaction parameterized by solution size.
\newblock volume~43, pages 573--616. SIAM, 2014.

\bibitem[Bod94]{bodlaender1994tourist}
Hans~L Bodlaender.
\newblock A tourist guide through treewidth.
\newblock {\em Acta cybernetica}, 11(1-2):1, 1994.

\bibitem[BRFGL10]{binkele2010exact}
Daniel Binkele~Raible, Henning Fernau, Serge Gaspers, and Mathieu Liedloff.
\newblock Exact exponential-time algorithms for finding bicliques.
\newblock {\em Information Processing Letters}, 111(2):64--67, 2010.

\bibitem[CFK{\etalchar{+}}16]{ParameterizedAlgorithm}
Marek Cygan, Fedor~V. Fomin, Lukasz Kowalik, Daniel Lokshtanov, D{\'a}niel
  Marx, Marcin Pilipczuk, Michal Pilipczuk, and Saket Saurabh.
\newblock {\em Parameterized Algorithms}.
\newblock Springer International Publishing, 1 edition, 2016.

\bibitem[CGG06]{chen2006parameterized}
Yijia Chen, Martin Grohe, and Magdalena Gr{\"u}ber.
\newblock On parameterized approximability.
\newblock In {\em Parameterized and Exact Computation}, pages 109--120.
  Springer, 2006.

\bibitem[CH06]{cai2006fixed}
Liming Cai and Xiuzhen Huang.
\newblock Fixed-parameter approximation: conceptual framework and
  approximability results.
\newblock In {\em International Workshop on Parameterized and Exact
  Computation}, pages 96--108. Springer, 2006.

\bibitem[CHKX04]{chen2004linear}
Jianer Chen, Xiuzhen Huang, Iyad~A Kanj, and Ge~Xia.
\newblock Linear {FPT} reductions and computational lower bounds.
\newblock In {\em Proceedings of the thirty-sixth annual ACM symposium on
  Theory of computing}, pages 212--221. ACM, 2004.

\bibitem[CK12]{couturier2012bicolored}
Jean~Fran{\c{c}}ois Couturier and Dieter Kratsch.
\newblock Bicolored independent sets and bicliques.
\newblock {\em Information Processing Letters}, 112(8):329--334, 2012.

\bibitem[CL16]{Chen2016TheCI}
Yijia Chen and Bingkai Lin.
\newblock The constant inapproximability of the parameterized dominating set
  problem.
\newblock {\em 2016 IEEE 57th Annual Symposium on Foundations of Computer
  Science (FOCS)}, pages 505--514, 2016.

\bibitem[DF99]{dowfel}
Rodney~G Downey and Michael~R Fellows.
\newblock {\em Parameterized Complexity}.
\newblock Springer-Verlag, 1999.

\bibitem[DF13]{fundpara2013}
Rodney~G Downey and Michael~R Fellows.
\newblock {\em Fundamentals of parameterized complexity}, volume~4.
\newblock Springer, 2013.

\bibitem[DFM06]{downey2006parameterized}
Rodney~G Downey, Michael~R Fellows, and Catherine McCartin.
\newblock Parameterized approximation problems.
\newblock In {\em International Workshop on Parameterized and Exact
  Computation}, pages 121--129. Springer, 2006.

\bibitem[Erd34]{erdos}
Paul Erd\H{o}s.
\newblock A {T}heorem of {S}ylvester and {S}chur.
\newblock {\em Journal London Mathematical Society}, s1-9 (4):278--282, 1934.

\bibitem[Erd59]{erdos1959graph}
Paul Erd\H{o}s.
\newblock Graph theory and probability.
\newblock {\em canad. J. Math}, 11:34--38, 1959.

\bibitem[Fei02]{feige2002relations}
Uriel Feige.
\newblock Relations between average case complexity and approximation
  complexity.
\newblock In {\em Proceedings of the thiry-fourth annual ACM symposium on
  Theory of computing}, pages 534--543. ACM, 2002.

\bibitem[FG06]{FlumGrohe06}
J{\"o}rg Flum and Martin Grohe.
\newblock {\em Parameterized Complexity Theory (Texts in Theoretical Computer
  Science. An EATCS Series)}.
\newblock Springer Verlag, Berlin, 2006.

\bibitem[FK04]{feige2004hardness}
Uriel Feige and Shimon Kogan.
\newblock Hardness of approximation of the balanced complete bipartite subgraph
  problem.
\newblock {\em Dept. Comput. Sci. Appl. Math., Weizmann Inst. Sci., Rehovot,
  Israel, Tech. Rep. MCS04-04}, 2004.

\bibitem[GKL12]{gaspers2012independent}
Serge Gaspers, Dieter Kratsch, and Mathieu Liedloff.
\newblock On independent sets and bicliques in graphs.
\newblock {\em Algorithmica}, 62(3-4):637--658, 2012.

\bibitem[Gro07]{Grohe07}
Martin Grohe.
\newblock The complexity of homomorphism and constraint satisfaction problems
  seen from the other side.
\newblock {\em Journal of the ACM (JACM)}, 54(1):1, 2007.

\bibitem[IP01]{impagliazzo2001complexity}
Russell Impagliazzo and Ramamohan Paturi.
\newblock On the complexity of k-sat.
\newblock {\em Journal of Computer and System Sciences}, 62:367--375, 2001.

\bibitem[IPZ98]{impagliazzo1998problems}
Russell Impagliazzo, Ramamohan Paturi, and Francis Zane.
\newblock Which problems have strongly exponential complexity?
\newblock In {\em Foundations of Computer Science, 1998. Proceedings. 39th
  Annual Symposium on}, pages 653--662. IEEE, 1998.

\bibitem[Joh87]{johnson1987np}
David~S Johnson.
\newblock The {NP}-completeness column: An ongoing guide.
\newblock {\em Journal of Algorithms}, 8(3):438--448, 1987.

\bibitem[Kho06]{khot2006ruling}
Subhash Khot.
\newblock Ruling out ptas for graph min-bisection, dense k-subgraph, and
  bipartite clique.
\newblock {\em SIAM Journal on Computing}, 36(4):1025--1071, 2006.

\bibitem[Klo94]{kloks1994treewidth}
Ton Kloks.
\newblock Treewidth, volume 842 of lecture notes in computer science, 1994.

\bibitem[KRS96]{kollar1996norm}
J{\'a}nos Koll{\'a}r, Lajos R{\'o}nyai, and Tibor Szab{\'o}.
\newblock Norm-graphs and bipartite {T}ur{\'a}n numbers.
\newblock {\em Combinatorica}, 16(3):399--406, 1996.

\bibitem[Kut12]{kutzkov2012exact}
Konstantin Kutzkov.
\newblock An exact exponential time algorithm for counting bipartite cliques.
\newblock {\em Information Processing Letters}, 112(13):535--539, 2012.

\bibitem[LN97]{lidl1997finite}
Rudolf Lidl and Harald Niederreiter.
\newblock {\em Finite fields}, volume~20.
\newblock Cambridge university press, 1997.

\bibitem[Mar07]{marx2007can}
D{\'a}niel Marx.
\newblock Can you beat treewidth?
\newblock In {\em Foundations of Computer Science, 2007. FOCS'07. 48th Annual
  IEEE Symposium on}, pages 169--179. IEEE, 2007.

\bibitem[Mar08]{marx2008parameterized}
D{\'a}niel Marx.
\newblock Parameterized complexity and approximation algorithms.
\newblock {\em The Computer Journal}, 51(1):60--78, 2008.

\bibitem[Ram19]{ram}
Srinivasa Ramanujan.
\newblock A proof of bertrand¡¦s postulate.
\newblock {\em Journal of the Indian Mathematical Society}, 11(181-182):27,
  1919.

\bibitem[RS86]{robertson1986graph}
Neil Robertson and Paul~D. Seymour.
\newblock Graph minors. ii. algorithmic aspects of tree-width.
\newblock {\em Journal of algorithms}, 7(3):309--322, 1986.

\bibitem[Sch76]{Schmidt}
Wolfgang~M. Schmidt.
\newblock {\em Equations over Finite Fields An Elementary Approach(Lecture
  Notes in Mathematics Volume 536)}.
\newblock Springer Berlin Heidelberg, 1976.

\bibitem[Shp13]{shparlinski2013finite}
Igor Shparlinski.
\newblock {\em Finite Fields: Theory and Computation: The meeting point of
  number theory, computer science, coding theory and cryptography}, volume 477.
\newblock Springer Science \& Business Media, 2013.

\bibitem[Xav12]{xavier2012note}
Eduardo~C Xavier.
\newblock A note on a maximum k-subset intersection problem.
\newblock {\em Information Processing Letters}, 112(12):471--472, 2012.

\end{thebibliography}
